\documentclass[11pt]{article}
\usepackage[letterpaper,margin=1in]{geometry}
\usepackage{cmap}
\usepackage[T1]{fontenc}
\usepackage[utf8]{inputenc}
\usepackage[english]{babel}
\usepackage{lmodern}
\usepackage{amsmath,amssymb,amsthm,mathtools}
\usepackage{microtype}
\usepackage{booktabs}
\usepackage{enumitem}
\usepackage{xcolor}
\usepackage{xspace}
\usepackage{mathtools}

\usepackage{interval}
\intervalconfig{soft open fences}

\newcommand{\braket}[2]{\left< #1 \vphantom{#2} \middle| #2 \vphantom{#1} \right>} 
\newcommand{\ketbra}[2]{\ensuremath{\ket{#1}\bra{#2}}}

\DeclarePairedDelimiter\rbra{\lparen}{\rparen}
\DeclarePairedDelimiter\sbra{\lbrack}{\rbrack}
\DeclarePairedDelimiter\cbra{\{}{\}}
\DeclarePairedDelimiter\abs{\lvert}{\rvert}
\DeclarePairedDelimiter\Abs{\lVert}{\rVert}
\DeclarePairedDelimiter\ceil{\lceil}{\rceil}

\let\ket\relax
\DeclarePairedDelimiter\ket{\lvert}{\rangle}
\let\bra\relax
\DeclarePairedDelimiter\bra{\langle}{\rvert}

\newcommand{\Tr} {\operatorname{Tr}}

\newcommand{\supp} {\operatorname{supp}}

\usepackage{amsmath,amssymb,amsthm,amsfonts,latexsym,mathtools,mathdots,bbm,graphicx,float}
\usepackage{caption,subcaption,ellipsis,xcolor,textcomp,thmtools,thm-restate,mleftright}
\usepackage{algorithm}
\usepackage{algpseudocode}
\usepackage{needspace,etoolbox}
\usepackage[pagebackref,colorlinks=true,citecolor=blue]{hyperref}
\usepackage[capitalize,nameinlink]{cleveref}
\newtheorem{theorem}{Theorem}[section]

\newtheorem{lemma}[theorem]{Lemma}

\newtheorem{fact}[theorem]{Fact}

\newtheorem{definition}[theorem]{Definition}

\newcommand{\reserveResultSpace}[1]{%
  \BeforeBeginEnvironment{#1}{\Needspace{4\baselineskip}}%
}
\forcsvlist{\reserveResultSpace}{theorem,axiom,lemma,corollary,proposition,fact,claim,definition,example,remark,problem}

\crefname{ineq}{inequality}{inequalities}
\creflabelformat{ineq}{#2{\upshape(#1)}#3}
\crefname{algorithm}{Algorithm}{Algorithms}
\crefname{axiom}{Axiom}{Axioms}
\Crefname{axiom}{Axiom}{Axioms}
\crefname{claim}{Claim}{Claims}
\Crefname{claim}{Claim}{Claims}
\crefname{fact}{Fact}{Facts}
\Crefname{fact}{Fact}{Facts}
\crefname{problem}{Problem}{Problems}
\Crefname{problem}{Problem}{Problems}
\usepackage{enumitem}

\newcommand{\ignore}[1]{}

\newcommand{\C}{\mathbb C}

\newcommand{\Id}{\operatorname{Id}}

\newcommand{\calC}{\mathcal{C}}
\newcommand{\calD}{\mathcal{D}}
\newcommand{\calE}{\mathcal{E}}

\newcommand{\calH}{\mathcal{H}}

\newcommand{\calL}{\mathcal{L}}

\newcommand{\calP}{\mathcal{P}}

\DeclarePairedDelimiterX\diverg[2]{(}{)}{#1 \,\|\, #2}
\newcommand{\dtv}[2]{\mathrm{d}_{\mathrm{TV}}(#1,#2)}

\newcommand{\dAa}[4]{\mathrm{A}_{#1}^{#4}\rbra*{#2,#3}}
\newcommand{\affh}[2]{\dAa{\mathrm H}{#1}{#2}{}}

\usepackage{algorithm}
\usepackage{algpseudocode}

\usepackage{tabularx}
\usepackage{booktabs}
\usepackage{threeparttable}

\usepackage{adjustbox}

\usepackage{footnotehyper} 
\makesavenoteenv{table}

\usepackage{tikz}

\newcommand{\opket}[1]{\lvert #1\rangle\!\rangle}
\newcommand{\opbra}[1]{\langle\!\langle #1\rvert}

\title{Optimal Tolerant Testing of Lindbladian Dissipation}
\author{
  Jinge Bao\thanks{University of Edinburgh. \href{mailto:jinge.bao@ed.ac.uk}{jinge.bao@ed.ac.uk}}
  \and
  Yiyi Cai\thanks{Stanford University. \href{mailto:yiyicai@stanford.edu}{yiyicai@stanford.edu}}
  \and
  Zhong-Xia Shang\thanks{University of Copenhagen. \href{mailto:zhongxia.shang@math.ku.dk}{zhongxia.shang@math.ku.dk}}
}

\date{}

\begin{document}

\maketitle

\begin{abstract}
Quantifying dissipation is essential for controlling quantum noise and characterizing open-system dynamics.
In practical experimental settings, residual noise may still persist despite efforts to suppress it, making it important to determine whether the dissipative strength remains within a prescribed tolerance threshold.
We study this question for an unknown, time-independent Lindblad generator with bounded strength, where the jump operators are local but with unrestricted overlap, using only memoryless forward-evolution access. 
The task is to distinguish dissipative strength of at most $\varepsilon_1$ from strength of at least $\varepsilon_2$, for $0 \leq \varepsilon_1 < \varepsilon_2$, measured in the canonical dissipator's normalized Frobenius norm.
We give an algorithm that solves this task in total evolution time $O\rbra{\varepsilon_2/\rbra{\varepsilon_2-\varepsilon_1}^2}$ for all nontrivial thresholds, with constant success probability, and provide a matching lower bound $\Omega\rbra{\varepsilon_2/\rbra{\varepsilon_2-\varepsilon_1}^2}$ that establishes optimality even for adaptive protocols. This extends dissipation testing to the tolerant setting and eliminates the bounded-degree requirement, making the framework applicable to a broader range of experimentally relevant settings.

\end{abstract}

\newpage

\setcounter{tocdepth}{2}
\tableofcontents
\clearpage

\section{Introduction}
\label{sec:introduction}
Quantum systems interact with their surroundings, exchanging energy and becoming correlated with their environments.
After tracing out the environment, the resulting reduced dynamics can exhibit dissipative effects such as energy relaxation and dephasing.
These effects can degrade coherence and entanglement \cite{PalmaSuominenEkert1996}, making the control of dissipation a central challenge in quantum information processing. 
Such unwanted coupling gives rise to the development of quantum error correction \cite{Shor1995} and fault-tolerant computation under suitable assumptions on physical noise~\cite{AliferisGottesmanPreskill2006}.
On the other hand, environmental coupling can also be used as a resource. 
Engineered dissipation can prepare and stabilize desired states \cite{PoyatosCiracZoller1996,DiehlEtAl2008}, including experimentally
demonstrated entangled states~\cite{BarreiroEtAl2011,ShankarEtAl2013},
and support universal quantum computation~\cite{VerstraeteWolfCirac2009}.
More recently, these open system dynamics have been used as tools for Gibbs sampling~\cite{RallWangWocjan2023,ChenKastoryanoGilyen2023,DingLiLin2024,ChenKastoryanoBrandaoGilyen2025,RouzeFrancaAlhambra2026,Ding2026Simple},
ground-state preparation~\cite{DingChenLin2024,Zhan2026Rapid}, and scientific computing tasks~\cite{ShangGuoAnZhao2025,Shang2026LinearSystemSolver}.
 Whether dissipation arises as unwanted noise or is deliberately engineered as a resource, it is important to determine how strongly it contributes to the observed dynamics.
 This motivates the problem of efficiently detecting and quantifying dissipation from experimental access to the system's time evolution. 


For finite-dimensional, time-homogeneous Markovian systems, coherent Hamiltonian evolution and dissipation are described within the Gorini--\allowbreak Kossakowski--\allowbreak Sudarshan--\allowbreak Lindblad framework~\cite{GoriniKossakowskiSudarshan1976,Lindblad1976}.
Characterizing these dynamics is challenging when both components are unknown, as both influence the observed evolution, so a change in the system's state does not by itself reveal the presence or strength of dissipation. 
Assessing dissipative effects therefore requires distinguishing them from coherent evolution, even when the underlying Hamiltonian and environmental couplings are unknown.
Lindbladian learning approaches this problem by reconstructing Hamiltonian and dissipative information from observed states or time evolution~\cite{RomanovEtAl2026,IvashkovEtAl2026,AradEtAl2026,MobusBergamaschiFrancaRouze2026,LewisTangWright2026}.
For many characterization tasks, however, the quantity of interest is much simpler than a complete model. 
When evaluating noise suppression, for example, one may wish to determine whether the remaining dissipative effects are sufficiently weak, without identifying every process that contributes to them. 
This motivates testing a property of the generator directly. 
We quantify dissipative strength by the normalized Frobenius norm of its canonical dissipator, which is uniquely determined by the generator. 
Equivalently, this quantity measures the generator’s distance from purely Hamiltonian dynamics, giving a precise meaning to the strength of its non-Hamiltonian component.

A natural starting point is to distinguish purely Hamiltonian dynamics from dynamics with some substantial dissipation \cite{Cai2026}.
This formulation addresses whether dissipation is present, but requires its exact absence in the acceptance case. 
This condition can be too restrictive when the aim is to assess residual noise.
Efforts to suppress environmental coupling may leave weak dissipative effects, and detecting their presence alone does not determine whether their strength lies within a prescribed tolerance. 
This motivates a tolerant formulation that accepts sufficiently weak dissipation and rejects sufficiently strong dissipation. 
More precisely, the test distinguishes strength at most a lower threshold from strength at least a higher threshold, with either answer allowed in between. 
This follows the framework of tolerant property testing \cite{ParnasRonRubinfeld2006}, which allows controlled deviations from an ideal property. 
Here, the central question is then how much evolution time is needed to make this distinction, and how that cost depends on the two thresholds.




In this work, we resolve this question using only memoryless forward-evolution access. 
We give the first algorithm for tolerant dissipation testing and establish its optimal interaction time through a matching lower bound. 
Furthermore, we only require each jump operator to act on at most a fixed number of qubits and remove the bounded-degree assumption of prior work \cite{Cai2026}.
Their supports and coefficients are unknown, and the Hamiltonian may be arbitrary and nonlocal. 
The resulting bounds have optimal dependence on both thresholds of the dissipation strength and are independent of system size. 
By allowing both weak background dissipation and unrestricted overlap among the local jump operators, our results extend dissipation testing to a broader range of experimentally relevant noise models.

\subsection{Main results}
\label{sec:main-results}
Consider an unknown, time-independent Lindblad generator $\calL$ on
$n\geq1$ qubits.
To define dissipative strength independently of the choice of jump operators, we use a representation with $\Tr\rbra{L_a} = 0$ and write
\begin{equation*}
    \calL\rbra{\rho} = -i\sbra{H,\rho}
    + \underset{\calD\rbra{\rho}}{\underbrace{\sum_a\rbra*{
        L_a\rho L_a^\dagger
        - \frac12\cbra*{L_a^\dagger L_a,\rho}
    }}}.
\end{equation*}
Here, $H$ governs coherent Hamiltonian evolution, while $\calD$ describes the dissipative contribution. 
The tracelessness condition ensures that $\calD$ is uniquely determined by $\calL$, even though the individual jump operators need not be unique. 
We therefore call $\calD$ the \emph{canonical dissipator}~\cite{HallCresserLiAndersson2014}, and we measure dissipative strength by its normalized Frobenius norm $\Abs{\calD}_F$.
We assume that $\calD$ admits a representation in which each jump operator acts on at most $k \geq 1$ qubits, without any requirement for the degree of the underlying interaction graph. 


The tester probes these unknown dynamics through the forward-evolution channels $\calE_t = e^{t\calL}$. 
Additionally, we enforce the bound $\Abs{\calL}_{\diamond} \leq \Lambda$, which controls errors when short evolutions are used to extract information about the generator.
Each call applies $\calE_t$ for a classically chosen duration $t \geq 0$, which may depend on earlier measurement outcomes. 
The tester may also use ancillae, retain relevant outputs, and perform joint measurements. 
We quantify the cost of probing the dynamics by the total interaction time, the sum of all query durations, including unsuccessful preparations and parallel calls.
The model permits arbitrarily short calls and does not charge for known operations or measurements.
Optimality here concerns total interaction time, not query count, gate count, storage, or wall-clock time.
Memorylessness describes the oracle, not a restriction on the tester's quantum memory.

\begin{theorem}[Time-optimal dissipation testing, informal version of \cref{thm:lo-time,thm:lo-time-lower}]
\label{thm:informal-time}
Consider the class of time-independent, $n$-qubit Lindblad generators $\calL = -i\sbra{H,\cdot} + \calD$ such that $\Abs{\calL}_{\diamond} \leq \Lambda$ and the canonical dissipator $\calD$ admits a representation by traceless jump operators, each supported on at most a fixed number $k$ of qubits.
Then, for $0 \leq \varepsilon_1 < \varepsilon_2 \leq \Lambda/\sqrt{2}$, there exists a randomized procedure which, given black-box access to the evolution channels $\cbra{\calE_t:= e^{t\calL}}_{t\geq 0}$, distinguishes between $\Abs{\calD}_{F} \leq \varepsilon_1$ and $\Abs{\calD}_{F} \geq \varepsilon_2$ with success probability at least $2/3$, using
\[
  O\rbra*{\frac{\varepsilon_2}{\rbra*{\varepsilon_2-\varepsilon_1}^2}}
\]
total evolution time.
Furthermore, a matching lower bound holds even for adaptive protocols with ancillary systems and joint measurements, establishing optimality of the protocol. 
\end{theorem}
For a fixed upper threshold, resolving a smaller gap requires quadratically more evolution time. 
When $\varepsilon_1=0$, our protocol also recovers the optimal
inverse-threshold Heisenberg-limited time scaling established in~\cite{Cai2026}, while removing the bounded-degree restriction.
Furthermore, the stated threshold range covers all nontrivial cases: by \cref{lem:generator-strength-bounds}, $\Abs{\calD}_F \leq \Abs{\calL}_{\diamond}/\sqrt{2} \leq \Lambda/\sqrt{2}$.

To achieve the upper bound, we express the dissipative strength in terms of two quantities: an overall jump rate $\Gamma$, which represents the total strength of the jumps on the maximally mixed input, and a normalized jump state $\sigma$ associated to $\Gamma$, which captures its structure independently of the overall scale. 
For nonzero dissipation, $\sigma$ encodes the action of the jump operators on one half of a maximally entangled system-reference register. 
Writing $d=2^n$ and $\sigma_{\mathsf R}$ for its marginal on the reference register, we define the \emph{jump-state functional} as
\begin{equation}
    h\rbra{\sigma}
  =\Tr\rbra*{\sigma^2}
   +\frac12\rbra*{d\Tr\rbra*{\sigma_{\mathsf R}^2}-1}. 
\end{equation}
This quantity combines the purity of the jump state with a measure of how far its reference marginal deviates from the maximally mixed state. 
As shown later, estimating this functional, together with the overall jump rate, suffices to determine the dissipative strength. 
This motivates studying the number of independent copies needed to estimate $h(\sigma)$ up to some accuracy and obtain the following optimal bound. 



\begin{theorem}[Copy-optimal jump-state functional estimation, informal version of \cref{thm:jump-state-estimation,thm:jump-state-copy-lower}]
\label{thm:informal-copy}
Let $0<\eta\leq1/16$.
Given independent copies of the jump state $\sigma$ of a nonzero
canonical dissipator with constant locality,
the optimal number of copies needed to estimate $h\rbra{\sigma}$
to additive accuracy $\eta$ with success probability at least $2/3$ is
$\Theta\rbra*{\eta^{-2}}$.
\end{theorem}

Ordinary purity estimation already has copy complexity independent
of dimension and proportional to the inverse square of the additive accuracy
\cite{EkertEtAl2002,ChenLiuWang2026TracePowers}.
This does not directly give the same bound for the jump-state
functional, whose marginal-purity term is multiplied by the system dimension.
Applying the standard disjoint-pair SWAP estimator separately to
the two purities gives a copy upper bound quadratic in the
system dimension.
Our estimator uses locality to obtain a dimension-independent
bound without known jump operators or supports.
The matching lower bound allows adaptive and collective measurements
and holds even for a single jump operator acting on one qubit.

\subsection{Technical overview}
\label{sec:technical-overview}


Our approach separates dissipative strength into an overall jump rate and a quantity describing the structure of the jumps. 
We first explain how to isolate this information from an unknown Hamiltonian, then show how to estimate it with a cost independent of system size. 
Finally, we show how to prepare the required states and control approximation errors within a fixed total evolution time. 

\paragraph{Isolating the jump contribution.}
We begin by isolating the jump contribution without reconstructing the unknown Hamiltonian. 
The Choi representation makes this possible by expressing the generator as an operator on the system and a reference register.
Let $\ket{\Phi}$ be a maximally entangled state of the system and a reference register, and we apply the generator $\calL$ to one half of the state, which gives us  
\begin{align}
    J(\calL) = (\calL \otimes I_{\mathsf R}) (\ket{\Phi}\bra{\Phi}).
\end{align}
In this representation, the Hamiltonian and anticommutator terms contain $\ket{\Phi}$ on at least one side, whereas each jump term is a positive operator built from the vector $(L_a \otimes I)\ket{\Phi}$.
As the jump operators are traveless, these vectors are orthogonal to $\ket{\Phi}.$

This distinction allows us to separate the contributions. 
Set $\Pi = I - \ket{\Phi}\bra{\Phi}$, which directly removes the Hamiltonian and anticommutator terms while preserving the jump terms:
\begin{align}
C = \Pi J(\calL) \Pi = \sum_a (L_a \otimes I) \ket{\Phi}\bra{\Phi} (L_a^{\dagger} \otimes I).
\end{align}
The resulting positive operator $C$ contains the information about the jumps needed for our test. 
We separate its overall scale from its normalized form.
Its trace $\Gamma = \Tr(C)$ is the total jump rate on the maximally mixed input. 
When $\Gamma > 0$, the normalized operator $\sigma = C/\Gamma$ is the jump state.
These quantities determine the dissipative strength through the exact identity 
\begin{align}
    \|\calD\|_F = \Gamma \sqrt{1+h(\sigma)}.
\end{align}
Our testing problem therefore reduces to estimating $\Gamma$ and the scalar function $h(\sigma)$.
If $\Gamma = 0$, the dissipator is zero and no jump-state estimation is needed. 

\paragraph{Estimating the jump-state functional.}
Suppose first that independent copies of $\sigma$ are available. 
Recall that 
\[h(\sigma) = \Tr(\sigma^2) + \frac{1}{2}(d \Tr(\sigma_{\mathsf R}^2 ) - 1).\]
While the global purity $\Tr(\sigma^2)$ can be estimated using ordinary two-copy SWAP measurements, the marginal purity is more obscure, as the contribution is multiplied by the system dimension $d$.
Estimating this purity directly would incur a dimension-dependent copy cost. 

We instead expand the reference marginal in the Pauli basis. 
Write $\alpha_P = \Tr(P\sigma_{\mathsf R})$ for some Pauli $P$, its contribution then becomes 
\begin{equation}
    G:= d \Tr(\sigma_{\mathsf R}^2) - 1 = \sum_{P \neq I} \alpha_P^2.
\end{equation}
Locality ensures that only Pauli strings of weight at most $k$ can contribute. 
However, there may nevertheless be many such strings, and their supports are unknown. 
Our estimator therefore first finds a small set of Pauli coefficients that captures most of $G$, then estimates their combined contribution. 
Measure $m$ copies in the Pauli-generated Bell basis, we get each outcome is a Pauli string acting on at most $k$ qubits. 
Since the reference marginal involves products of Pauli components of the jumps, we construct candidate strings from pairwise unions of the observed supports.
Each union contains at most $2k$ qubits; for constant $k$, this gives $O(m^2)$ candidates that is independent of system size. 

The candidate set need not contain every nonzero coefficient.
Locality guarantees that the expected sum of squared coefficients omitted is $O(1/m)$, so $m = O(\eta^{-1})$ samples suffice to make this contribution $O(\eta)$ with constant probability. 
Once the candidate set is identified, we need to estimate its total contribution to $G$.
Recovering each coefficient accurately would provide more information than we need. 
We instead estimate the sum of their squares directly, which allows the same measurement data to contribute to the entire sum.

Locality is essential here as well as in candidate discovery, as it controls the statistical fluctuations of this aggregate estimate, even when the jump supports overlap arbitrarily. 
Consequently, the estimation cost depends on the locality bound and the desired accuracy, without growing with system size.
Together with the approximation guarantee for the candidate set, this yields an $O_k(\eta^{-2})$ copy bound for the marginal contribution. 
Global purity can then be estimated within the same copy count, which completes the estimator for $h(\sigma)$.

\paragraph{Preparation of jump states.}
We now need to obtain copies of jump state $\sigma$ using only forward evolution.
Fortunately, the projection that isolates the jump contribution also provides a way to prepare these copies.
After evolving the system half of a maximally entangled probe for a short time $t$, projecting onto the subspace orthogonal to the initial probe leaves, to first order, 
\[t C  = \Gamma t \sigma.\]
This expression describes both the probability of obtaining the projection outcome and the state retained afterward.
In the ideal first-order experiment, the outcome occurs with probability $\Gamma t$, and its normalized output is exactly $\sigma$.
Thus, this experiment supplies both quantities needed by the tester. 
The frequency of successful preparations reveals the jump rate, while the retained states reveal the jump-state functional.
By choosing sufficiently short evolutions, we control the accumulated approximation error within a fixed total evolution time.

\paragraph{From copy complexity to evolution time.}
The jump rate already gives a coarse estimate of dissipative strength. 
Locality implies 
\[
\Gamma \leq \|\calD\|_F \leq q_k \Gamma,
\]
where $q_k$ only depends on $k$.
The tester uses these bounds to decide whether the rate estimate alone suffices or whether the jump-state functional must also be estimated.

If the estimated rate is below $\epsilon_2/(2q_k)$, the upper bound on dissipative strength rules out the high-dissipation case, so the tester accepts. If the estimate exceeds $2\epsilon_2$, the lower bound rules out the low-dissipation case, so the tester rejects.
These cutoffs include margins for estimation error and give the correct answer on promised inputs whenever the rate estimate is sufficiently accurate.

Functional estimation is needed only when the estimated rate lies between these cutoffs. 
In this case, an accurate estimate guarantees that $\Gamma$ is comparable to $\epsilon_2$ up to constants depending on $k$. 
This is essential for controlling the preparation cost. 
Whenever jump-state copies are needed, their production rate is sufficiently large. 
If successful preparations are too rare, the rate alone already settles the decision.
Resolving the threshold gap requires functional accuracy proportional to $(\epsilon_2 - \epsilon_1)/\epsilon_2$.
Combining the inverse-square copy requirement with the rate at which they are obtained gives a total evolution time of 
\[O\left(\frac{\epsilon_2}{(\epsilon_2 - \epsilon_1)^2}\right).\]

\paragraph{Matching lower bounds.}
The copy and evolution time lower bounds use different constructions. 
For copy estimation, we choose nearby pure jump states, each generated by a single jump operator on one qubit.
Their global purities are identical, but their reference-marginal contributions produce functional values separated by more than twice the required accuracy.
Distinguishing these nearby states requires $\Omega(\eta^{-2})$ copies, even with adjoint measurements.
The accuracy dependence is therefore unavoidable even in this simple setting.

For evolution time, single-qubit dephasing generators realize the two promise thresholds.
heir forward channels can be simulated by a Poisson process of phase flips, with a rate proportional to dissipative strength. 
Every adaptive tester can be simulated using such a record, so its distinguishing power is bounded by that of the full record. 
Comparing the two Poisson laws gives
\[T = \Omega\left(\frac{\epsilon_2}{(\epsilon_2 - \epsilon_1)^2}\right).\]
This lower bound applies directly to the evolution-access problem and establishes optimality beyond the particular preparation and estimation strategy used by our algorithm.
Both lower bounds already hold on a single qubit, which shows that the accuracy dependence reflects an intrinsic statistical difficulty of these tasks.

\subsection{Related work}
\label{sec:related-work}
Our work connects dissipation testing with quantum property testing, state-functional estimation, and Lindbladian learning. 
We review these connections below, highlighting how the testing objectives, structural assumptions, and access models compare with those considered here.

\paragraph{Testing quantum dynamics.}
Our work directly builds on generator-level dissipation detection \cite{Cai2026}, which distinguishes purely Hamiltonian dynamics from dissipation above a prescribed threshold.
Under locality and bounded-degree assumptions on the jump operators, that work establishes optimal evolution time despite an arbitrary unknown Hamiltonian. 
We extend the task to allow nonzero dissipation in the acceptance case and remove the bounded-degree restriction.

A related tolerant testing problem considers a known unitary followed by unknown depolarizing noise. 
In that setting, ~\cite[Theorem~3]{GhoshJacksonDatta2026} gives a channel-query upper bound with the same threshold dependence as our evolution time bound, together with a matching lower bound for individual measurements.
The similarity in scaling arises because both problems distinguish two regimes separated by a threshold, even though they use different testing models and distance measures: their thresholds are defined by the diamond distance from a known target channel, whereas ours are defined by the normalized Frobenius norm of the canonical dissipator.
We also allow an arbitrary unknown Hamiltonian and charge the duration of evolution calls. 
Other approaches to characterizing quantum dynamics study channel unitarity \cite{WallmanEtAl2015, ChenEtAl2023} and symmetry testing for channels and Lindbladians~\cite{BandyopadhyayEtAl2023}

\paragraph{Quantum property testing.}
Our two-threshold formulation places dissipation testing within the broader study of quantum property testing. 
For Hamiltonians, this literature considers general properties \cite{BluhmCaroOufkir2026}, locality \cite{KallaugherLiang2025}, Pauli sparsity \cite{ArunachalamDuttEscuderoGutierrez2026}, and certification against a target \cite{GaoJiWangYuZhao2026}. 
For quantum operations, testing problems include membership in unitary subgroups \cite{BrakerskiSharmaWeissenberg2021}, unitary juntas \cite{ChenLiLuo2024, BaoLiuYaoYeZhang2025}, Clifford operations \cite{HinscheEtAl2025}, and shallow channels \cite{DongOuYao2025}. 
All these problems share the same target of deciding the presence of some structural or quantitative property without reconstructing the full object.

The state-testing literature addresses questions for juntas ~\cite{BEG26} and stabilizer states \cite{GrewalIyerKretschmerLiang2023,ArunachalamDutt2024,BaoVanDordrechtHelsen2024,MehrabanTahmasbi2024,IyerLiang2024}. 
Other work studies product-state approximation \cite{BakshiEtAl2024}, fermionic
and bosonic Gaussianity
\cite{BittelMeleEisertLeone2024,LeoneBittel2026,GirardiEtAl2025},
certification to a reference state~\cite{BaoGaoWang2025,Wang2026Fidelity},
and quantum-state closeness testing with purified query access~\cite{GilyenLi2019}.
The access models and resources vary across these settings, including state copies, black-box queries, and evolution time. 
In comparison, our problem only concerns total forward interaction time for an unknown Lindblad generator.



\paragraph{State-functional estimation.}
Our upper bound relies on estimating a nonlinear functional of a quantum state from independent copies.
Relevant results include direct estimation of linear and nonlinear functionals ~\cite{EkertEtAl2002}, trace-power estimation
\cite{LiuWang2026TracePowers,ChenLiuWang2026TracePowers},
and simultaneous estimation of nonlinear
functionals~\cite{ChenWangYuZhang2026}.
Related statistical tools include methods for predicting state properties\cite{HuangKuengPreskill2020} and quantum
U-statistics~\cite{GutaButucea2010}.

The distinguishing feature of our functional is its dimension-rescaled marginal-purity term. 
Applying standard purity estimators directly to this term introduces a dimension-dependent copy cost. 
We exploit the locality of the underlying jump operators to obtain a dimension-independent bound. 
The estimator combines product-Pauli measurement records through a classical U-statistic, with locality controlling both the contribution omitted during candidate discovery and the variance of the resulting estimate.

\paragraph{Lindbladian learning.}
Lindbladian learning aims to recover information about the coefficients of the generator from observed states or dynamics.
General approaches include channel and generator tomography \cite{ChuangNielsen1997,BenderskyPastawskiPaz2008,BoulantEtAl2003}, robust and experimental Lindbladian estimation
\cite{SamachEtAl2022,LiuSeddonKohlerOnoratiCubitt2025,BirkeEtAl2026}, and learning open-system dynamics from steady states \cite{BaireyEtAl2020}.
 These methods provide a detailed description of the dynamics; our task asks only whether the canonical dissipator’s strength lies below or above specified thresholds.

Recent work develops learning methods for local quantum dynamics \cite{FrancaMarkovichEtAl2024,AradEtAl2026,SongZhangYuanWu2026}, time-dependent dynamics and noise~\cite{FrancaEtAl2025},
and unknown generator structure
\cite{MobusBergamaschiFrancaRouze2026,LewisTangWright2026}.
Further approaches relax the need for prescribed supports\cite{IvashkovEtAl2026,RomanovEtAl2026,ZhouGong2026,ChenYu2026,Sinha2026}., while numerical and practical aspects are studied in~\cite{HeightmanEtAl2026,LamJoshiFrancaVermersch2026}.
These results differ in their structural assumptions, access models, and reconstruction guarantees. 
Our contribution focuses on the interaction time required for a scalar testing task under the locality restriction alone.

There is also a direct connection between the lower-bound constructions. An adaptive learning-time lower bound based on nearby dephasing generators appears in ~\cite[Theorem~4.2 and Corollary~4.3]{AradEtAl2026}. 
We use the same generator family to establish the optimal dependence on the two thresholds in tolerant dissipation testing.



\paragraph{Organization}
\Cref{sec:preliminaries} introduces the notation and background
on Lindblad dynamics and canonical dissipation.
\Cref{sec:testing-strategy-overview} describes the access model and explains how jump-state estimation leads to a bound on total evolution time. 
\Cref{sec:jump-state-estimation} establishes matching copy-complexity bounds for the jump-state functional,  and \cref{sec:locality-only} establishes matching copy-complexity bounds for the jump-state functional.

\section{Preliminaries} \label{sec:preliminaries}

\Cref{sec:notations} fixes notation, and
\cref{sec:quantum-information} reviews the quantum-information background.
\Cref{sec:canonical-dissipation} introduces Lindblad dynamics,
the canonical dissipator, the jump state, and the jump-state functional.
Finally, \cref{sec:gauge-invariant-dissipation} discusses
gauge invariance and the interpretation of dissipative
strength, and introduces the locality condition
on jump operators.

\subsection{Notation} \label{sec:notations}
Let $n\geq1$ be an integer, and set
$\calH=\rbra{\C^2}^{\otimes n}$ and $d=2^n$.
Let $\mathsf{L}\rbra{\calH}$ denote the space of linear operators
on $\calH$, and let $I$ denote the identity operator on the
relevant space.
Denote the identity map on $\mathsf{L}\rbra{\calH}$ by
$\Id_{\calH}$, so that $\Id_{\calH}\rbra{X}=X$ for every
$X\in\mathsf{L}\rbra{\calH}$.
Denote the computational basis of $\calH$ by
$\cbra{\ket{j}}_{j=0}^{d-1}$ and the matrix entries of
$X\in\mathsf{L}\rbra{\calH}$ in this basis by
$X_{ij}=\bra{i}X\ket{j}$ for $i,j\in\cbra{0,\ldots,d-1}$.
For an operator $X$, write $X^\dagger$ for its adjoint and
$X^{\mathsf T}$ for its transpose in the computational basis.
For a scalar $z\in\C$, denote its complex conjugate by $z^*$.
For $X,Y\in\mathsf{L}\rbra{\calH}$, denote their commutator
and anticommutator by $\sbra{X,Y}=XY-YX$ and
$\cbra{X,Y}=XY+YX$, respectively.
Let $\Tr$ denote the trace.
For $X\in\mathsf{L}\rbra{\calH}$, denote the normalized trace of $X$ by
$\overline{\Tr}\rbra{X}=d^{-1}\Tr\rbra{X}$,
so that $\overline{\Tr}\rbra{I}=1$.
For $X,Y\in\mathsf{L}\rbra{\calH}$, the Hilbert--Schmidt
inner product and its normalized version are
$\Tr\rbra{X^\dagger Y}$ and
$\overline{\Tr}\rbra{X^\dagger Y}$, respectively.
For a matrix $X$ and $1\leq p<\infty$, denote its Schatten
$p$-norm by
\begin{align*}
    \Abs{X}_p
    = \rbra*{\Tr\rbra*{\rbra*{X^\dagger X}^{p/2}}}^{1/p}.
\end{align*}
The cases $p=1$ and $p=2$ give the trace norm and the
unnormalized Hilbert--Schmidt norm, respectively.
For vectors, $\Abs{\cdot}_2$ denotes the Euclidean norm.
Let $\Abs{X}_\infty$ denote the operator norm of $X$,
which is its largest singular value.
Denote by $\mathbf 1\cbra{\cdot}$ the indicator function,
which equals $1$ when the condition in braces holds
and $0$ otherwise.
Denote by $\log$ the natural logarithm.

Throughout, the locality bound $k$ is a fixed positive integer,
and implicit constants in asymptotic bounds may depend only on $k$.

For two probability distributions on the same measurable space,
denote by $\dtv{\cdot}{\cdot}$ their total variation distance,
the supremum of the absolute difference in their probabilities
over measurable events.
For discrete distributions with probability mass functions $p$ and $q$
on a common finite or countable set, their total variation distance
and Hellinger affinity are
\[
    \dtv{p}{q}
    = \frac12\sum_x\abs*{p\rbra{x}-q\rbra{x}},
    \qquad
    \affh{p}{q} = \sum_x\sqrt{p\rbra{x}q\rbra{x}},
\]
respectively, where both sums run over the common outcome set.

Let $\calP_n=\cbra{I,X,Y,Z}^{\otimes n}$ be the set of
Pauli strings, where $X,Y,Z$ are the single-qubit Pauli
matrices, and set $\calP_n^\circ=\calP_n\setminus\cbra{I}$.
For $P=P_1\otimes\cdots\otimes P_n\in\calP_n$, denote
its support and weight by
$\supp\rbra*{P}=\cbra*{j\in\cbra*{1,\ldots,n}:P_j\neq I}$
and $\operatorname{wt}\rbra{P}=\abs*{\supp\rbra*{P}}$,
respectively.
For a set of sites $S\subseteq\cbra{1,\ldots,n}$,
denote its complement by
$S^c=\cbra{1,\ldots,n}\setminus S$.
For $P,Q\in\calP_n$, Pauli orthogonality gives
$\Tr\rbra{\rbra{P/\sqrt d}^\dagger\rbra{Q/\sqrt d}}
    = \overline{\Tr}\rbra{P^\dagger Q}
    = \delta_{P,Q}$.
Since $\abs{\calP_n}=4^n=d^2=\dim\mathsf{L}\rbra{\calH}$,
the set $\cbra{P/\sqrt d}_{P\in\calP_n}$ is an orthonormal
basis of $\mathsf{L}\rbra{\calH}$ with respect to the
Hilbert--Schmidt inner product.

Let $\mathsf S$ be the system register and $\mathsf R$ an isomorphic
reference register, with the system factor placed first in tensor products.
Denote the partial traces over these registers by
$\Tr_{\mathsf S}$ and $\Tr_{\mathsf R}$, respectively,
and the identity map on the reference register by $\Id_{\mathsf R}$.

\subsection{Basics of quantum information} \label{sec:quantum-information}

The normalized maximally entangled vector of
$\mathsf S$ and $\mathsf R$ is
\begin{align*}
    \ket{\Phi} = \frac{1}{\sqrt d}\sum_{j=0}^{d-1}\ket{j}_{\mathsf S}\ket{j}_{\mathsf R}.
\end{align*}
Its reduced state on either register is $I/d$.
For $X\in\mathsf{L}\rbra{\calH}$, denote its vectorization by
\begin{align}
    \opket{X}
    = \rbra*{X\otimes I}\ket{\Phi}
    = \frac{1}{\sqrt d}\sum_{i,j=0}^{d-1}
      X_{ij}\ket{i}_{\mathsf S}\ket{j}_{\mathsf R}.
\end{align}
For all $X,Y\in\mathsf{L}\rbra{\calH}$, orthonormality
of the computational basis therefore gives
\begin{align}
    \opbra{X}\opket{Y}
    = \frac{1}{d}\sum_{j=0}^{d-1}\rbra{X^\dagger Y}_{jj}
    = \frac{1}{d}\Tr\rbra{X^\dagger Y}
     = \overline{\Tr}\rbra{X^\dagger Y}.
\end{align}
The inner-product identity and Pauli orthogonality imply that
$\opbra{P}\opket{Q}=\delta_{P,Q}$ for $P,Q\in\calP_n$.
Since $\abs{\calP_n}=d^2=\dim\rbra{\calH\otimes\calH}$,
$\cbra{\opket{P}}_{P\in\calP_n}$ is an orthonormal basis
of $\calH\otimes\calH$.
This is the Pauli-vector basis, with $\opket{I}=\ket{\Phi}$.
The vectors indexed by $\calP_n^\circ$ span $\ket{\Phi}^\perp$.

\begin{fact}[Partial-trace identity]
\label{fact:vectorized-partial-trace}
For all $X,Y\in\mathsf{L}\rbra{\calH}$,
\begin{align}
    \Tr_{\mathsf S}\rbra*{\opket{X}\opbra{Y}}
    = \frac{1}{d}\rbra*{Y^\dagger X}^{\mathsf T}.
    \label{eq:vectorized-partial-trace}
\end{align}
\end{fact}

\begin{proof}
For $j,\ell\in\cbra{0,\ldots,d-1}$, tracing over $\mathsf S$
in the computational-basis expansion above contracts the system
indices, giving
\begin{align}
    \bra{j}_{\mathsf R}
    \Tr_{\mathsf S}\rbra*{\opket{X}\opbra{Y}}
    \ket{\ell}_{\mathsf R}
    = \frac{1}{d}\sum_{i=0}^{d-1}
       X_{ij}\rbra{Y^\dagger}_{\ell i}
    = \frac{1}{d}\rbra*{Y^\dagger X}_{\ell j}
     = \frac{1}{d}
       \rbra*{\rbra*{Y^\dagger X}^{\mathsf T}}_{j\ell}.
\end{align}
\end{proof}
For a linear map
$\calE\colon\mathsf{L}\rbra{\calH}\to\mathsf{L}\rbra{\calH}$,
denote its normalized Choi operator~\cite{Choi1975} by
\begin{align}
    J\rbra*{\calE}
    = \rbra*{\calE\otimes\Id_{\mathsf R}}
       \rbra*{\ketbra{\Phi}{\Phi}}
    = \frac{1}{d}\sum_{i,j=0}^{d-1}
       \calE\rbra*{\ketbra{i}{j}}\otimes\ketbra{i}{j}.
\end{align}
The map can be recovered from its Choi operator.
For every $X\in\mathsf{L}\rbra{\calH}$ and
$i,j\in\cbra{0,\ldots,d-1}$,
\begin{align*}
    \Tr\rbra*{\ketbra{i}{j}X^{\mathsf T}}
    = \bra{j}X^{\mathsf T}\ket{i}
    = X_{ij}.
\end{align*}
Substituting this identity into the Choi expansion gives
\begin{align}
    d\Tr_{\mathsf R}\rbra*{
        J\rbra*{\calE}\rbra*{I\otimes X^{\mathsf T}}}
    = \sum_{i,j=0}^{d-1}
       X_{ij}\calE\rbra*{\ketbra{i}{j}}
    = \calE\rbra*{
         \sum_{i,j=0}^{d-1}X_{ij}\ketbra{i}{j}}
    = \calE\rbra*{X}.
\end{align}
For a linear map
$\calE\colon\mathsf{L}\rbra{\calH}\to\mathsf{L}\rbra{\calH}$,
let $\Abs*{\calE}_{\mathrm{HS}}$ denote the
Hilbert--Schmidt norm of its matrix representation in a
Hilbert--Schmidt orthonormal operator basis.
Equivalently, if $\cbra{E_a}_{a=1}^{d^2}$ is any operator basis satisfying $\Tr\rbra{E_a^\dagger E_b}=\delta_{ab}$, then 
\begin{align*}
    \Abs*{\calE}_{\mathrm{HS}}^2 = \sum_{a=1}^{d^2} \Abs*{\calE\rbra*{E_a}}_2^2.
\end{align*}
Denote the normalized Frobenius norm by
\begin{align*}
    \Abs*{\calE}_F^2
    = \frac{1}{d^2}\Abs*{\calE}_{\mathrm{HS}}^2  
    = \frac{1}{d^2}\sum_{a=1}^{d^2}\Abs*{\calE\rbra*{E_a}}_2^2.
\end{align*}
This normalization gives
\begin{align*}
    \Abs*{\Id_{\calH}}_F^2
    = \frac{1}{d^2}\sum_{a=1}^{d^2}\Abs*{E_a}_2^2
    = \frac{1}{d^2}\sum_{a=1}^{d^2}1
     = 1.
\end{align*}

For linear maps $\calE_1,\calE_2$ on
$\mathsf{L}\rbra{\calH}$, denote their normalized
Frobenius inner product by
\begin{align}
    \langle\calE_1,\calE_2\rangle_F
    = \frac{1}{d^2}\sum_{a=1}^{d^2}
      \Tr\rbra*{
          \calE_1\rbra{E_a}^\dagger
          \calE_2\rbra{E_a}}.
\end{align}
This inner product is independent of the chosen
Hilbert--Schmidt orthonormal operator basis and satisfies
$\langle\calE,\calE\rangle_F=\Abs*{\calE}_F^2$.

\begin{fact}[Choi--Frobenius identity]
\label{fact:choi-frobenius}
For linear maps
$\calE_1,\calE_2\colon\mathsf{L}\rbra{\calH}\to\mathsf{L}\rbra{\calH}$,
\begin{align*}
    \langle\calE_1,\calE_2\rangle_F
    = \Tr\rbra*{
        J\rbra{\calE_1}^\dagger J\rbra{\calE_2}}.
\end{align*}
In particular, for every such map $\calE$,
\begin{align}
    \Abs*{\calE}_F = \Abs*{J\rbra*{\calE}}_2.
    \label{eq:choi-frobenius-norm}
\end{align}
\end{fact}

\begin{proof}
The matrix units $\cbra{\ketbra{i}{j}}_{i,j=0}^{d-1}$
form a Hilbert--Schmidt orthonormal operator basis, since
\begin{align}
    \Tr\rbra*{\rbra*{\ketbra{i}{j}}^\dagger
                    \ketbra{k}{\ell}}
    = \braket{i}{k}\braket{\ell}{j}
     = \delta_{ik}\delta_{j\ell}.
\end{align}
Applying this orthogonality to the reference factors in the
two Choi expansions gives
\begin{align}
    \Tr\rbra*{J\rbra{\calE_1}^\dagger J\rbra{\calE_2}}
    = \frac{1}{d^2}\sum_{i,j=0}^{d-1}
      \Tr\rbra*{
          \calE_1\rbra*{\ketbra{i}{j}}^\dagger
          \calE_2\rbra*{\ketbra{i}{j}}}
    = \langle\calE_1,\calE_2\rangle_F.
\end{align}
Setting $\calE_1=\calE_2=\calE$ and taking nonnegative
square roots gives the norm identity.
\end{proof}
Denote the diamond norm of $\calE$ by
\begin{align}
    \Abs*{\calE}_\diamond
    = \sup_{0\neq X\in\mathsf{L}\rbra*{\calH\otimes\calH}}
      \frac{
        \Abs*{\rbra*{\calE\otimes\Id_{\mathsf R}}
        \rbra*{X}}_1
      }{\Abs*{X}_1}.
\end{align}
Throughout, $\Abs*{\cdot}_F$ quantifies the magnitude of dissipation,
while $\Abs*{\cdot}_{\diamond}$ controls the strength of the generator
and finite-time approximation errors.

\subsection{Lindbladian dynamics} \label{sec:canonical-dissipation}

Let $H\in\mathsf{L}\rbra{\calH}$ be Hermitian, and let
$\cbra{L_a}_a\subseteq\mathsf{L}\rbra{\calH}$ be a finite
family of jump operators.
The associated dissipator acts on $X\in\mathsf{L}\rbra{\calH}$ as
\begin{align}
    \calD\rbra*{X} = \sum_a L_a X L_a^\dagger - \frac{1}{2}\cbra*{A,X},
    \qquad
    A = \sum_a L_a^\dagger L_a.
\end{align}
The corresponding time-independent Lindblad generator
is~\cite{GoriniKossakowskiSudarshan1976,Lindblad1976}
\begin{align}
    \calL\rbra*{X} = -i\sbra*{H,X} + \calD\rbra*{X}.
\end{align}
The maps $\calE_t=e^{t\calL}$, $t\geq0$, form a
one-parameter semigroup of quantum channels.

The Lindblad representation is not unique.
Subtracting the identity component from each jump operator
and adjusting $H$ accordingly yields a representation with
traceless jump operators without changing
$\calL$~\cite{HallCresserLiAndersson2014}.
Henceforth, let $H,\cbra{L_a}_a,A$, and $\calD$ denote the
corresponding quantities in a chosen representation with traceless
jump operators, so that $\Tr\rbra*{L_a}=0$ for every $a$.


The positive semidefinite operator $A$ determines the total
instantaneous jump rate. For an input state $\rho$, cyclicity
of the trace gives
\begin{align}
    \sum_a \Tr\rbra*{L_a\rho L_a^\dagger}
    = \sum_a \Tr\rbra*{L_a^\dagger L_a\rho}
    = \Tr\rbra*{A\rho}.
\end{align}
The vectorized jump operators give the positive semidefinite
jump block on $\calH\otimes\calH$,
\begin{align*}
    C = \sum_a \opket{L_a}\opbra{L_a}.
\end{align*}
Denote its trace by $\Gamma$.
The vectorization inner-product identity gives
\begin{align}
    \Gamma
    = \Tr\rbra*{C}
    = \sum_a \opbra{L_a}\opket{L_a}
    = \sum_a \overline{\Tr}\rbra*{L_a^\dagger L_a}
    = \overline{\Tr}\rbra*{A}.
\end{align}
Thus $\Gamma$ is the total jump rate for the maximally mixed
input $I/d$.
Let $\Pi=I-\ketbra{\Phi}{\Phi}$ be the orthogonal projector
onto $\ket{\Phi}^{\perp}$.
Tracelessness gives
$\bra{\Phi}\opket{L_a}=\overline{\Tr}\rbra{L_a}=0$,
so $\Pi\opket{L_a}=\opket{L_a}$ for every $a$.
Consequently,
\begin{align}
    \Pi C\Pi
    = \sum_a \Pi\opket{L_a}\opbra{L_a}\Pi
    = \sum_a \opket{L_a}\opbra{L_a}
    = C.
\end{align}

\begin{lemma}[Canonical jump block, adapted from {\cite[Appendix~C, Eq.~(C8)]{HallCresserLiAndersson2014}}]
\label{lem:canonical-dissipator}
Let $\calL$ be an $n$-qubit Lindblad generator with
traceless jump operators $\cbra{L_a}_a$.
The associated jump block $C$ satisfies
\begin{align*}
    C = \Pi J\rbra*{\calL}\Pi.
\end{align*}
\end{lemma}

\begin{proof}
For each $a$, normalized vectorization gives
\begin{align*}
    \rbra*{L_a\otimes I}\ketbra{\Phi}{\Phi}
    \rbra*{L_a^\dagger\otimes I}
    = \opket{L_a}\opbra{L_a}.
\end{align*}
Since $H$ and $A$ are Hermitian, for $X\in\cbra{H,A}$,
\begin{align}
    \rbra*{X\otimes I}\ketbra{\Phi}{\Phi}
    = \opket{X}\bra{\Phi},
    \qquad
    \ketbra{\Phi}{\Phi}\rbra*{X\otimes I}
    = \ket{\Phi}\opbra{X}.
\end{align}
Substituting these identities into
$J\rbra{\calL}
=\rbra{\calL\otimes\Id_{\mathsf R}}\rbra*{\ketbra{\Phi}{\Phi}}$
yields
\begin{align*}
    J\rbra*{\calL}
    = C
      -i\opket{H}\bra{\Phi}
      +i\ket{\Phi}\opbra{H}
      -\frac12\opket{A}\bra{\Phi}
      -\frac12\ket{\Phi}\opbra{A}.
\end{align*}
Since $\Pi\ket{\Phi}=0$ and $\bra{\Phi}\Pi=0$, all terms
involving $H$ or $A$ are zero after projection by $\Pi$
on both sides. Using $\Pi C\Pi=C$ then gives
\begin{align*}
    \Pi J\rbra*{\calL}\Pi
    = \Pi C\Pi
    = C.
\end{align*}
\end{proof}

By \cref{lem:canonical-dissipator}, $C$ is uniquely determined
by $\calL$.
Since $C$ is the Choi operator of the jump map
$X\mapsto\sum_a L_aXL_a^\dagger$, injectivity of the Choi
representation implies that $C$ determines this map.
Moreover, \cref{fact:vectorized-partial-trace} gives
\begin{align}
    \Tr_{\mathsf S}\rbra*{C}
    = \sum_a \Tr_{\mathsf S}\rbra*{\opket{L_a}\opbra{L_a}}
    = \frac{1}{d}\sum_a\rbra*{L_a^\dagger L_a}^{\mathsf T}
    = \frac{A^{\mathsf T}}{d}.
\end{align}
Thus $A=d\rbra*{\Tr_{\mathsf S}\rbra*{C}}^{\mathsf T}$
is also determined by $C$.
The jump map and $A$ together determine $\calD$.
Consequently, $\calD$ is uniquely determined by $\calL$
and is called its \emph{canonical dissipator}.
The total jump rate $\Gamma=\Tr\rbra{C}$ and, when
$\Gamma>0$, the normalized jump block $C/\Gamma$ are
also independent of the chosen representation with traceless jump operators.
The latter is the jump state introduced below.

\begin{definition}[Jump state] \label{def:jump-state}
Let $\calD$ be an $n$-qubit canonical dissipator with
jump block $C$ and total jump rate $\Gamma$.
For $\Gamma>0$, the jump state of $\calD$ is
\begin{align*}
    \sigma = \frac{C}{\Gamma}.
\end{align*}
\end{definition}

For $\Gamma>0$, $\sigma=C/\Gamma$ is positive semidefinite
and satisfies
\begin{align}
    \Tr\rbra{\sigma}
    = \frac{\Tr\rbra{C}}{\Gamma}
    = 1,
    \qquad
    \sigma\ket{\Phi}
    = \frac{C\ket{\Phi}}{\Gamma}
    = 0.
\end{align}
Thus $\sigma$ is a quantum state supported on
$\ket{\Phi}^{\perp}$.
By \cref{fact:vectorized-partial-trace}, its reference
marginal satisfies
\begin{align} \label{eq:jump-state-marginal}
    \sigma_{\mathsf R}
    = \Tr_{\mathsf S}\rbra*{\sigma}
    = \frac{1}{\Gamma}\Tr_{\mathsf S}\rbra*{C}
    = \frac{A^{\mathsf T}}{d\Gamma}.
\end{align}
If $\Gamma=0$, then
\begin{align}
    \Gamma
    = \frac{1}{d}\sum_a\Tr\rbra*{L_a^\dagger L_a}
    = \frac{1}{d}\sum_a\Abs{L_a}_2^2
    = 0.
\end{align}
Since each summand is nonnegative, $\Abs{L_a}_2=0$
and hence $L_a=0$ for every $a$.
Consequently, $\calD=0$.
No jump state is assigned in this case.

\begin{definition}[Jump-state functional]
\label{def:jump-functional}
Let $\sigma$ be the jump state of a nonzero canonical
dissipator on $n$ qubits.
The \emph{jump-state functional} $h$ is given by
\begin{align} \label{eq:jump-functional}
    h\rbra*{\sigma}
    = \Tr\rbra*{\sigma^2}
      + \frac{1}{2}\rbra*{d\Tr\rbra*{\sigma_{\mathsf R}^2}-1}.
\end{align}
\end{definition}

The following lemma expresses the normalized Frobenius norm
of $\calD$ in terms of its total jump rate $\Gamma$ and
the jump-state functional $h\rbra{\sigma}$.

\begin{lemma}[Rate--state identity]
\label{lem:rate-state-identity}
Let $\calD$ be a nonzero $n$-qubit canonical dissipator
with total jump rate $\Gamma>0$ and jump state $\sigma$.
Then
\begin{align*}
    \Abs*{\calD}_F = \Gamma\sqrt{1+h\rbra*{\sigma}}.
\end{align*}
\end{lemma}

\begin{proof}
The Choi operator of $\calD$ is
\begin{align}
    J\rbra*{\calD}
    = C-\frac12\opket{A}\bra{\Phi}
       -\frac12\ket{\Phi}\opbra{A}.
\end{align}
Since $C=C^\dagger$, $J\rbra{\calD}$ is Hermitian.
The identities $C\ket{\Phi}=0$ and $\bra{\Phi}C=0$
imply that all cross terms involving $C$ have zero trace.
Since $A=A^\dagger$, normalized vectorization gives
\begin{align}
    \bra{\Phi}\opket{A}
    = \opbra{A}\ket{\Phi}
    = \overline{\Tr}\rbra{A}
    = \Gamma,
    \qquad
    \opbra{A}\opket{A}
    = \overline{\Tr}\rbra{A^2}.
\end{align}
By \cref{fact:choi-frobenius}, expanding the square
and taking the trace gives
\begin{align}
    \Abs*{\calD}_F^2
    &= \Tr\rbra*{J\rbra{\calD}^2}\\
    &= \Tr\rbra*{C^2}
       +\frac14\rbra*{
          \Gamma^2+\overline{\Tr}\rbra{A^2}
          +\overline{\Tr}\rbra{A^2}+\Gamma^2}\\
    &= \Tr\rbra*{C^2}
       +\frac12\Gamma^2
       +\frac12\overline{\Tr}\rbra{A^2}.
\end{align}
The relation $C=\Gamma\sigma$ gives
$\Tr\rbra{C^2}=\Gamma^2\Tr\rbra{\sigma^2}$.
Moreover, $A=d\Gamma\sigma_{\mathsf R}^{\mathsf T}$
and invariance of the trace under transposition give
\begin{align}
    \overline{\Tr}\rbra*{A^2}
    = \frac1d\Tr\rbra*{
        \rbra*{d\Gamma\sigma_{\mathsf R}^{\mathsf T}}^2}
    = d\Gamma^2\Tr\rbra*{\sigma_{\mathsf R}^2}.
\end{align}
Substituting these expressions yields
\begin{align}
    \Abs*{\calD}_F^2
    = \Gamma^2\rbra*{
        \Tr\rbra*{\sigma^2}
        +\frac12+\frac d2\Tr\rbra*{\sigma_{\mathsf R}^2}}
    = \Gamma^2\rbra*{1+h\rbra*{\sigma}}.
\end{align}
Taking nonnegative square roots and using $\Gamma>0$
completes the proof.
\end{proof}

\begin{lemma}[Strength bounds]
\label{lem:generator-strength-bounds}
Let $\calL$ be an $n$-qubit Lindblad generator, and
let $\calD$ be its canonical dissipator with total
jump rate $\Gamma$.
Then
\begin{align*}
    \Gamma \leq \Abs*{\calL}_\diamond,
    \qquad
    \Abs*{\calD}_F
    \leq \frac{\Abs*{\calL}_\diamond}{\sqrt2}.
\end{align*}
\end{lemma}

\begin{proof}
By \cref{lem:canonical-dissipator},
$C=\Pi J\rbra{\calL}\Pi\geq0$.
Compression by the orthogonal projector $\Pi$ does not
increase the trace norm. Thus
\begin{align}
    \Gamma
    = \Tr\rbra{C}
    = \Abs{C}_1
    = \Abs*{\Pi J\rbra{\calL}\Pi}_1
    \leq \Abs*{J\rbra{\calL}}_1
    \leq \Abs*{\calL}_\diamond.
\end{align}
The last inequality follows from the normalized Choi
representation and the diamond-norm bound applied to
$\ketbra{\Phi}{\Phi}$, whose trace norm is one.

For the second bound, the Hamiltonian and dissipative
parts have Hermitian Choi operators
\begin{align}
    J\rbra*{-i\sbra*{H,\cdot}}
    = -i\opket{H}\bra{\Phi}
       +i\ket{\Phi}\opbra{H},
    \qquad
    J\rbra*{\calD}
    = C-\frac12\opket{A}\bra{\Phi}
       -\frac12\ket{\Phi}\opbra{A}.
\end{align}
Using $C\ket{\Phi}=0$, $\bra{\Phi}C=0$, and
$\overline{\Tr}\rbra{A}=\Gamma$, the trace of their
product is
\begin{align}
    \Tr\rbra*{
        J\rbra*{-i\sbra*{H,\cdot}}J\rbra{\calD}}
    &= \frac{i}{2}\rbra*{
        \Gamma\overline{\Tr}\rbra{H}
        +\overline{\Tr}\rbra{AH}
        -\overline{\Tr}\rbra{HA}
        -\Gamma\overline{\Tr}\rbra{H}}\\
    &= \frac{i}{2}\rbra*{
        \overline{\Tr}\rbra{AH}
        -\overline{\Tr}\rbra{HA}}\\
    &= 0,
\end{align}
where the last equality follows from cyclicity of the trace.
By \cref{fact:choi-frobenius}, this gives
$\langle\calD,-i\sbra*{H,\cdot}\rangle_F=0$.
Consequently,
\begin{align}
    \Abs*{J\rbra{\calL}}_2^2
    = \Abs*{J\rbra*{-i\sbra*{H,\cdot}}}_2^2
      +\Abs*{J\rbra{\calD}}_2^2
    \geq \Abs*{J\rbra{\calD}}_2^2
    = \Abs*{\calD}_F^2.
\end{align}
Moreover, $J\rbra{\calL}$ is Hermitian.
The Lindblad form and cyclicity of the trace give
$\Tr\rbra*{\calL\rbra{X}}=0$ for every $X$, so
\begin{align*}
    \Tr\rbra*{J\rbra{\calL}}
    = \Tr\rbra*{\calL\rbra{I/d}}
    = 0.
\end{align*}
For any traceless Hermitian operator $X$, the positive
eigenvalues and the absolute values of the negative
eigenvalues each sum to $\Abs{X}_1/2$.
For each sign, the sum of the squared magnitudes is
at most the square of their sum. Hence
\begin{align}
    \Abs{X}_2^2
    \leq 2\rbra*{\frac{\Abs{X}_1}{2}}^2
    = \frac{\Abs{X}_1^2}{2}.
\end{align}
Applying this inequality to $J\rbra{\calL}$ gives
\begin{align}
    \Abs*{\calL}_\diamond
    \geq \Abs*{J\rbra{\calL}}_1
    \geq \sqrt2\Abs*{J\rbra{\calL}}_2
    \geq \sqrt2\Abs*{\calD}_F.
\end{align}
\end{proof}

\subsection{Dissipative strength} \label{sec:gauge-invariant-dissipation}

The Hamiltonian--dissipator decomposition of a Lindblad
generator generally depends on the chosen representation.
In the gauge with traceless jump operators, however, $C$, $A$, and the
canonical dissipator $\calD$ are uniquely determined by
$\calL$, as established in~\cref{sec:canonical-dissipation}.
Consequently, $\Abs*{\calD}_F$, the total jump rate
$\Gamma$, and, when $\Gamma>0$, the jump state $\sigma$
depend only on $\calL$, although the individual traceless
jump operators need not be unique.

The following lemma characterizes $\Abs*{\calD}_F$
as the normalized Frobenius distance from $\calL$ to
the space of Hamiltonian generators.

\begin{lemma}[Distance from Hamiltonian generators]
\label{lem:dissipation-distance-hamiltonian}
Let $\calL=-i\sbra{H,\cdot}+\calD$ be an $n$-qubit
Lindblad generator, where $H=H^\dagger$ and $\calD$
is its canonical dissipator.
Then
\begin{equation}
\label{eq:dissipation-distance-hamiltonian}
    \Abs*{\calD}_F
    = \min_{K=K^\dagger}
      \Abs*{\calL+i\sbra*{K,\cdot}}_F.
\end{equation}
The minimum is attained precisely when $K-H$ is
a real scalar multiple of $I$.
In particular, $\Abs*{\calD}_F=0$ if and only if
$\calL$ is a Hamiltonian generator.
\end{lemma}

\begin{proof}
For every Hermitian $K$, the Choi operators
\begin{align}
    J\rbra{\calD}
    = C-\frac12\opket{A}\bra{\Phi}
       -\frac12\ket{\Phi}\opbra{A},
    \qquad
    J\rbra*{-i\sbra*{K,\cdot}}
    = -i\opket{K}\bra{\Phi}
       +i\ket{\Phi}\opbra{K}
\end{align}
are Hermitian.
Using $C\ket{\Phi}=0$ and $\bra{\Phi}C=0$,
the terms involving $C$ have zero trace.
The two terms proportional to
$\Gamma\overline{\Tr}\rbra{K}$ cancel, leaving
\begin{align*}
    \Tr\rbra*{
        J\rbra{\calD}J\rbra*{-i\sbra*{K,\cdot}}}
    = \frac{i}{2}\rbra*{
        \overline{\Tr}\rbra{AK}
        -\overline{\Tr}\rbra{KA}}
    = 0,
\end{align*}
where the last equality follows from cyclicity of the trace.
By \cref{fact:choi-frobenius}, this gives
$\langle\calD,-i\sbra*{K,\cdot}\rangle_F=0$.
Since $H-K$ is Hermitian, applying this orthogonality
with $K$ replaced by $H-K$ gives
\begin{align}
    \Abs*{\calL+i\sbra*{K,\cdot}}_F^2
    = \Abs*{\calD-i\sbra*{H-K,\cdot}}_F^2
    = \Abs*{\calD}_F^2
       +\Abs*{-i\sbra*{H-K,\cdot}}_F^2.
\end{align}
The last term is nonnegative and is zero precisely
when $H-K$ commutes with every matrix unit, or
equivalently, when $H-K$ is a real scalar multiple of $I$.
Taking $K=H$ therefore proves the distance formula
and the characterization of its minimizers.
Finally, the minimum is zero exactly when
$\calL=-i\sbra{K,\cdot}$ for some Hermitian $K$,
which proves the last assertion.
\end{proof}

The testing problem therefore concerns the normalized
Frobenius distance of $\calL$ from the space of Hamiltonian
generators.
Here \emph{dissipation} refers to non-Hamiltonian generator
strength in this norm.
No identification with thermodynamic heat flow or a
state-dependent entropy-production rate is assumed.

For an arbitrary Lindblad representation of $\calL$
with jump operators $\cbra{V_a}_a$, centering these operators as
$L_a=V_a-\overline{\Tr}\rbra{V_a}I$ and adjusting the
Hamiltonian accordingly leaves $\calL$ unchanged.
Expanding the centered products and using
$\overline{\Tr}\rbra{I}=1$ gives
\begin{align}
    \Gamma
    = \sum_a\overline{\Tr}\rbra{L_a^\dagger L_a}
    = \sum_a\rbra*{
        \overline{\Tr}\rbra{V_a^\dagger V_a}
        -\abs*{\overline{\Tr}\rbra{V_a}}^2}.
\end{align}
By contrast, the total jump rate for the maximally mixed
input $I/d$ in the uncentered representation is
$\sum_a\overline{\Tr}\rbra{V_a^\dagger V_a}$,
which generally depends on the chosen representation.

The following locality condition is adapted from~\cite{Cai2026}.
\begin{definition}[Local jump operators]
\label{def:local-jump-operators}
Let $k\geq1$ be an integer, and let $\cbra{L_a}_a$ be
a finite family of $n$-qubit jump operators.
For each $a$ and $P\in\calP_n$, denote the normalized
Pauli coefficient of $L_a$ corresponding to $P$ by
$\gamma_{a,P}=\overline{\Tr}\rbra{PL_a}$.
A jump operator $L_a$ is called $k$-local if there exists
a set $S_a\subseteq\cbra{1,\ldots,n}$ with
$\abs{S_a}\leq k$ such that $\gamma_{a,P}=0$ for every
$P\in\calP_n$ satisfying $\supp\rbra{P}\nsubseteq S_a$.
The family has locality at most $k$ if every $L_a$
is $k$-local.

A canonical dissipator has locality at most $k$ if it
can be represented by traceless $k$-local jump operators.
\end{definition}

Centering does not enlarge the support of a jump
operator. If $V_a=v_a\otimes I_{S_a^c}$, where $v_a$ acts
on the qubits in $S_a$, then
\begin{align*}
    \overline{\Tr}\rbra{V_a}
    = \frac{2^{n-\abs{S_a}}\Tr\rbra{v_a}}{2^n}
    = \frac{\Tr\rbra{v_a}}{2^{\abs{S_a}}}.
\end{align*}
The centered jump operator therefore satisfies
\begin{align}
    L_a
    = V_a-\overline{\Tr}\rbra{V_a}I
    = \rbra*{
        v_a-\frac{\Tr\rbra{v_a}}{2^{\abs{S_a}}}I_{S_a}}
      \otimes I_{S_a^c}.
\end{align}
Thus a representation with $k$-local jump operators yields a
canonical dissipator with locality at most $k$.
This locality condition requires the existence of a
representation with traceless $k$-local jump operators.
It does not require every Lindblad representation of
the same generator to have $k$-local jump operators.

Let $k\geq1$ be an integer, and let $\calL$ be a Lindblad
generator whose nonzero canonical dissipator has locality
at most $k$ and jump state $\sigma$.
For any $c>0$, replacing $H$ by $cH$ and each $L_a$ by
$\sqrt{c}\,L_a$ gives a representation of $c\calL$.
Its jump block and total jump rate are therefore $cC$ and
$c\Gamma$, respectively, so its jump state is
$cC/\rbra{c\Gamma}=\sigma$.
Since multiplication by a nonzero scalar preserves the
support of each jump operator, the locality is unchanged.
Moreover,
$\Abs{c\calL}_\diamond=c\Abs{\calL}_\diamond$.
Thus, for any prescribed $\Lambda>0$, choosing $c>0$
sufficiently small ensures
$\Abs{c\calL}_\diamond\leq\Lambda$ without changing $\sigma$
or the locality of the canonical dissipator.
Consequently, such an upper bound does not restrict the
class of jump states of canonical dissipators with
locality at most $k$.

The following bounds will be used in both functional estimation
and tolerant testing.
\begin{lemma}[Locality bounds for the jump state]
\label{lem:jump-state-locality-bounds}
Fix a locality bound $k\geq1$.
Let $\sigma$ be the jump state of a nonzero $n$-qubit canonical
dissipator with locality at most $k$, and let $d=2^n$.
Then
\[
  \sigma_{\mathsf R}\leq\frac{2^k}{d}I,
  \qquad
  1\leq d\Tr\rbra*{\sigma_{\mathsf R}^2}\leq2^k,
  \qquad
  0\leq h\rbra{\sigma}\leq\frac{2^k+1}{2}.
\]
\end{lemma}

\begin{proof}
Choose a representation with nonzero traceless jump operators
$L_a$ supported on sets $S_a$ with $\abs{S_a}\leq k$.
The state $\sigma$ is a convex combination of the normalized
jump-vector states
$\opket{L_a}\opbra{L_a}/\overline{\Tr}\rbra{L_a^\dagger L_a}$,
with weights $\overline{\Tr}\rbra{L_a^\dagger L_a}/\Gamma$.
The reference marginal of each such state is a state on $S_a$
tensored with the maximally mixed state on $S_a^c$.
Since a density operator is bounded above by the identity,
this marginal is bounded above by
\[
  I_{S_a}\otimes\frac{I_{S_a^c}}{2^{n-\abs{S_a}}}
  =\frac{2^{\abs{S_a}}}{d}I
  \leq\frac{2^k}{d}I.
\]
Convexity gives the first bound.
Cauchy--Schwarz and $\Tr\rbra*{\sigma_{\mathsf R}}=1$ give
$d\Tr\rbra*{\sigma_{\mathsf R}^2}\geq1$.
The operator bound gives
$d\Tr\rbra*{\sigma_{\mathsf R}^2}\leq
2^k\Tr\rbra*{\sigma_{\mathsf R}}=2^k$.
Finally, $0\leq\Tr\rbra*{\sigma^2}\leq1$ and
\cref{def:jump-functional} imply the stated bounds on $h\rbra{\sigma}$.
\end{proof}

\section{Testing framework}
\label{sec:testing-strategy-overview}

This section specifies the testing model and explains how
the copy estimator in \cref{sec:jump-state-estimation}
yields the time upper bound in \cref{sec:locality-only}.

Fix a locality bound $k\geq1$.
Let $n\geq1$ be an integer, let $\Lambda>0$ be a strength bound,
and let $0\leq\varepsilon_1<\varepsilon_2$ be thresholds.
Let $\calL$ be an unknown time-independent $n$-qubit
Lindblad generator with $\Abs{\calL}_\diamond\leq\Lambda$
and canonical dissipator $\calD$ of locality at most $k$.
The task is to distinguish
\[
  \mathrm{YES}:\ \Abs{\calD}_F\leq\varepsilon_1
  \qquad\text{from}\qquad
  \mathrm{NO}:\ \Abs{\calD}_F\geq\varepsilon_2
\]
with success probability at least $2/3$.
Either output is permitted when
$\varepsilon_1<\Abs{\calD}_F<\varepsilon_2$.

The tester is given $n,k,\Lambda,\varepsilon_1,\varepsilon_2$
and access to a memoryless forward-evolution oracle
for $\calL$.
For any classically chosen duration $t\geq0$, a call
applies $\calE_t=e^{t\calL}$ to an $n$-qubit register,
possibly entangled with ancillary systems and
outputs of earlier calls.

Consider adaptive testers whose number of oracle calls
has a finite upper bound depending only on the known
parameters.
Time is measured by the sum of the durations of all oracle calls.
For $T\geq0$, a tester uses total time at most $T$ if,
for all measurement outcomes and internal random choices,
its call durations $t_j$ satisfy $\sum_j t_j\leq T$.
The sum includes calls used in unsuccessful
state-preparation attempts and adds the durations
of parallel calls.
Known operations and measurements do not contribute to this time.

By \cref{lem:generator-strength-bounds},
$\Abs{\calD}_F\leq\Lambda/\sqrt2$.
If $\varepsilon_2>\Lambda/\sqrt2$, always returning
$\mathrm{YES}$ solves the problem with zero
total time.
The time bounds therefore assume
$\varepsilon_2\leq\Lambda/\sqrt2$.

For $\Gamma>0$, \cref{lem:rate-state-identity} expresses
dissipative strength in terms of $\Gamma$ and $h\rbra{\sigma}$.
The locality bound $0\leq h\rbra{\sigma}\leq\rbra{2^k+1}/2$
from \cref{lem:jump-state-locality-bounds} gives
\begin{align}
    \Gamma
  \leq\Abs{\calD}_F
  =\Gamma\sqrt{1+h\rbra{\sigma}}
  \leq\sqrt{1+\frac{2^k+1}{2}}\,\Gamma.
\end{align}
Thus a sufficiently small rate excludes the $\mathrm{NO}$
case, while a sufficiently large rate excludes the
$\mathrm{YES}$ case.
The tester estimates $h\rbra{\sigma}$ only when neither
rate-based decision applies.
When the rate estimate is accurate, this remaining case
has $\Gamma=\Theta\rbra{\varepsilon_2}$, as verified in
the proof of \cref{thm:lo-time}.
If $\Gamma=0$, then $\calD=0$, so no jump state needs
to be estimated.

Both quantities can be estimated using the following
short-time experiment.
Prepare $\ket{\Phi}$, evolve the system register for time
$t>0$, and perform the projective measurement
$\cbra{\ketbra{\Phi}{\Phi},\Pi}$.
On outcome $\Pi$, record a success and store the
postmeasurement state.
For $\Lambda t\leq1/2$, the ideal first-order experiment
in \cref{subsec:lo-time-upper} succeeds with probability
$\Gamma t$ and, when $\Gamma>0$, outputs exactly $\sigma$
conditional on success.
For an integer $M\geq1$, consider $M$ independent
repetitions of the ideal experiment.
Let $T=Mt$ and let $K$ be the number of successes.
Then $K\sim\operatorname{Bin}\rbra{M,\Gamma t}$.
The success count gives an unbiased rate estimate
$\widehat\Gamma=K/T$, with
\[
  \mathbb E\sbra{K}=\Gamma T,
  \qquad
  \operatorname{Var}\sbra*{\widehat\Gamma}
  =\frac{\Gamma\rbra{1-\Gamma t}}{T}
  \leq\frac{\Gamma}{T}.
\]
Conditional on the success flags, the successful outputs
are independent copies of $\sigma$ whenever $\Gamma>0$.
The same oracle calls therefore supply both the rate
estimate and the copies used for functional estimation.

To determine the required accuracy, suppose $\Gamma>0$
and let $\widehat h$ be the estimate from
\cref{thm:jump-state-estimation}.
This estimate satisfies $0\leq\widehat h\leq\rbra{2^k+1}/2$.
Using $\abs{\sqrt{1+x}-\sqrt{1+y}}\leq\abs{x-y}/2$
for $x,y\geq0$ gives
\begin{align}
  \abs*{\widehat\Gamma\sqrt{1+\widehat h}-\Abs{\calD}_F}
  &\leq\sqrt{1+\frac{2^k+1}{2}}\,
       \abs*{\widehat\Gamma-\Gamma}
    +\frac{\Gamma}{2}\abs*{\widehat h-h\rbra{\sigma}}.
\end{align}
In the intermediate-rate case, it therefore suffices
to estimate $\Gamma$ to additive accuracy proportional
to $\varepsilon_2-\varepsilon_1$ and to choose the
functional accuracy $\eta$ proportional to
$\rbra{\varepsilon_2-\varepsilon_1}/\varepsilon_2$.
With sufficiently small constants depending only on $k$,
the strength error is less than $\rbra{\varepsilon_2-\varepsilon_1}/2$,
so comparison with $\rbra{\varepsilon_1+\varepsilon_2}/2$
settles the promise problem.
By \cref{thm:jump-state-estimation}, the functional
estimate uses
\[
  O\rbra*{\eta^{-2}}
  =O\rbra*{\frac{\varepsilon_2^2}
    {\rbra*{\varepsilon_2-\varepsilon_1}^2}}
\]
copies.
The ideal expected yield $\Gamma T$ and
$\Gamma=\Theta\rbra{\varepsilon_2}$ motivate the total-time
scale
\[
  T=O\rbra*{\frac{1}{\Gamma\eta^2}}
   =O\rbra*{\frac{\varepsilon_2}
     {\rbra*{\varepsilon_2-\varepsilon_1}^2}}.
\]
Chebyshev's inequality and the variance bound above show
that choosing $T$ at this scale, with a sufficiently large
constant depending only on $k$, also gives the required
rate accuracy in the intermediate-rate case.

The construction in \cref{subsec:lo-time-upper} chooses $T$
and $M$ from the known parameters and runs exactly $M$
attempts of duration $t=T/M$, using total time $T$.
Its condition for invoking functional estimation ensures
that enough successful outputs have been stored.

\Cref{alg:tolerant-dissipation-testing} summarizes the tester in \cref{thm:lo-time}.

\begin{algorithm}[h]
\caption{Tolerant dissipation testing}
\label{alg:tolerant-dissipation-testing}
\begin{algorithmic}[1]
\Require Number of qubits $n$, locality bound $k$,
strength bound $\Lambda$, dissipation thresholds
$0\leq\varepsilon_1<\varepsilon_2\leq\Lambda/\sqrt2$,
and forward-evolution access to $\calL$.
\Ensure $\mathrm{YES}$ or $\mathrm{NO}$.
\State Set
$q_k=\sqrt{1+\rbra{2^k+1}/2}$ and
$\eta=\rbra{\varepsilon_2-\varepsilon_1}/\rbra{32\varepsilon_2}$.
\State Choose $T>0$ and an integer $M\geq1$, and set $t=T/M$.
\For{$M$ independent attempts}
  \State Prepare $\ket{\Phi}$ and apply $\calE_t$
  to the system register.
  \State Perform the projective measurement
  $\cbra{\ketbra{\Phi}{\Phi},\Pi}$,
  preserving coherence within each outcome subspace.
  \State On outcome $\Pi$, record a success
  and store the postmeasurement state.
\EndFor
\State Let $K$ be the number of successes
and set $\widehat\Gamma=K/T$.
\If{$\widehat\Gamma\leq\varepsilon_2/\rbra{2q_k}$}
  \State \Return $\mathrm{YES}$.
\ElsIf{$\widehat\Gamma\geq2\varepsilon_2$}
  \State \Return $\mathrm{NO}$.
\EndIf
\State Apply the estimator (\cref{thm:jump-state-estimation}) with parameter $\eta$ to stored outputs,
obtaining $\widehat h$.
\State \Return $\mathrm{YES}$ if
$\widehat\Gamma\sqrt{1+\widehat h}
\leq\rbra{\varepsilon_1+\varepsilon_2}/2$,
and $\mathrm{NO}$ otherwise.
\end{algorithmic}
\end{algorithm}

To control the finite-time approximation error, compare
the complete physical output of each preparation, including
its success or failure flag, with the corresponding ideal
output.
For $\Lambda t\leq1/2$, their trace distance is
$O\rbra{\Lambda^2t^2}$.
For $M$ independent attempts, the resulting change in any
final decision probability is at most
\[
  O\rbra{M\Lambda^2t^2}
  =O\rbra{\Lambda^2Tt}
  =O\rbra*{\frac{\Lambda^2T^2}{M}}.
\]
Increasing $M$ while keeping $T$ fixed makes this error
sufficiently small.
The copy guarantee is applied to exact ideal jump states,
and the comparison transfers the final success guarantee
to the physical experiment.
\Cref{subsec:lo-time-upper} gives the full construction
and error analysis.
The separate evolution-access lower bound in
\cref{subsec:lo-time-lower} establishes optimality of
the resulting time dependence for fixed $k$.

\section{Jump-state functional estimation} \label{sec:jump-state-estimation}

Motivated by \cref{sec:testing-strategy-overview}, this
section studies functional estimation from independent
copies, with an arbitrary accuracy parameter and without
counting the cost of preparing the copies.

Fix a locality bound $k\geq1$.
Let $n\geq1$ be an integer and let $0<\eta\leq1/16$
be an accuracy parameter.
Let $\sigma$ be the jump state of a nonzero $n$-qubit
canonical dissipator with locality at most $k$.
Given $n,k,\eta$ and independent copies of $\sigma$,
the task is to output an estimate $\widehat h$ of
$h\rbra{\sigma}$ such that, for every such $\sigma$,
\begin{align}
    \Pr\sbra*{\abs*{\widehat h-h\rbra{\sigma}}\leq\eta}
    \geq\frac{2}{3}.
\end{align}

\Cref{sec:pauli-support-discovery} constructs
a candidate set of Pauli strings.
\Cref{sec:copy-upper-bound} constructs an estimator
for the jump-state functional and establishes
its copy upper bound.
Finally, \cref{sec:copy-lower-bound} proves
a matching copy lower bound.

\subsection{Candidate Pauli strings}
\label{sec:pauli-support-discovery}

For each $P\in\calP_n^\circ$, denote the Pauli coefficient
of the reference marginal $\sigma_{\mathsf R}$ by $a_P=\Tr\rbra*{P\sigma_{\mathsf R}}$.
The coefficients $a_P$ are real because both $P$ and
$\sigma_{\mathsf R}$ are Hermitian.
Denote the sum of their squares by
\begin{align*}
    G=\sum_{P\in\calP_n^\circ}a_P^2.
\end{align*}
Expanding $\sigma_{\mathsf R}$ in the Pauli basis and
using $\Tr\rbra*{\sigma_{\mathsf R}}=1$ gives
\begin{align*}
    \sigma_{\mathsf R}
    =\frac{1}{d}\rbra*{
        I+\sum_{P\in\calP_n^\circ}a_PP
    }.
\end{align*}
Using $\Tr\rbra{P}=0$ for $P\in\calP_n^\circ$ and
$\Tr\rbra{PQ}=d\delta_{P,Q}$ yields
\begin{align*}
    d\Tr\rbra*{\sigma_{\mathsf R}^2}
    =1+\sum_{P\in\calP_n^\circ}a_P^2
    =1+G.
\end{align*}
Consequently,
\begin{align*}
    h\rbra{\sigma}
    =\Tr\rbra*{\sigma^2}
     +\frac12\rbra*{d\Tr\rbra*{\sigma_{\mathsf R}^2}-1}
    =\Tr\rbra*{\sigma^2}+\frac{G}{2}.
\end{align*}

Let $m\geq1$ be an integer.
Measure $m$ independent copies of $\sigma$ in the
Pauli-vector basis $\cbra{\opket{Q}}_{Q\in\calP_n}$,
obtaining labels $Q_1,\ldots,Q_m$.
Each vectorized $k$-local jump operator is a linear
combination of Pauli vectors indexed by strings of
weight at most $k$.
The jump state $\sigma$ is proportional to a sum of
outer products of such vectors, so
$\operatorname{wt}\rbra{Q_i}\leq k$ for all $i$ with
probability one.
Construct the candidate set
\begin{equation}
\label{eq:lo-candidates}
    \calC=
    \bigcup_{1\leq i,j\leq m}
    \cbra*{
        P\in\calP_n^\circ:
        \operatorname{wt}\rbra{P}\leq k,\,
        \supp\rbra{P}\subseteq
        \supp\rbra{Q_i}\cup\supp\rbra{Q_j}
    }.
\end{equation}
There are at most $m^2$ support unions.
Each contains at most $2k$ qubits and therefore
supports at most $4^{2k}=16^k$ Pauli strings.
Consequently,
\begin{equation}
\label{eq:lo-candidate-size}
    \abs{\calC}\leq16^k m^2.
\end{equation}
Denote the contribution to $G$ from the candidate set
$\calC$ by
\begin{align*}
    G_{\calC}=\sum_{P\in\calC}a_P^2.
\end{align*}

Thus $G-G_{\calC}$ is the sum of $a_P^2$ over
$P\in\calP_n^\circ\setminus\calC$.
The next lemma bounds its expectation.

\begin{lemma}[Expected approximation error]
\label{lem:lo-discovery}
For every integer $m\geq1$, the candidate set $\calC$ in
\eqref{eq:lo-candidates} satisfies
\begin{equation}
\label{eq:lo-discovery-loss}
    \mathbb E\sbra*{G-G_{\calC}}
    \leq\frac{2\cdot16^k\rbra{4^k-1}}{em}.
\end{equation}
The expectation is over the measurement outcomes
$Q_1,\ldots,Q_m$.
\end{lemma}

\begin{proof}
Fix a representation of the canonical dissipator with
nonzero traceless jump operators $L_a$ supported on sets
$S_a$ with $\abs{S_a}\leq k$.
Set
\begin{align*}
    w_a=\frac{\opbra{L_a}\opket{L_a}}{\Gamma},
    \qquad
    \ket{\psi_a}
    =\frac{\opket{L_a}}
           {\sqrt{\opbra{L_a}\opket{L_a}}}.
\end{align*}
Then $\sigma=\sum_a w_a\ketbra{\psi_a}{\psi_a}$,
$w_a>0$, and $\sum_a w_a=1$.
For $P\in\calP_n^\circ$, let
\begin{align*}
    w\rbra{P}=\sum_{a:\supp\rbra{P}\subseteq S_a}w_a.
\end{align*}
Each $S_a$ supports exactly $4^{\abs{S_a}}-1$
nonidentity Pauli strings.
Interchanging the sums and using $\abs{S_a}\leq k$
and $\sum_a w_a=1$ gives
\begin{equation}
\label{eq:lo-support-weight-sum}
    \sum_{P\in\calP_n^\circ}w\rbra{P}
    =\sum_a w_a\rbra*{4^{\abs{S_a}}-1}
    \leq4^k-1.
\end{equation}

Fix $P\in\calP_n^\circ$ with $R=\supp\rbra{P}$
and $\abs{R}\leq k$.
Write $P\rvert_R$ for its restriction to the sites in $R$,
ordered increasingly and indexed by their original labels.
Use the same restriction notation for other Pauli strings.
Since $L_a$ is supported on $S_a$, the reference marginal
of $\ket{\psi_a}$ has a maximally mixed tensor factor
at every site outside $S_a$.
If $R\nsubseteq S_a$, then $P$ has a traceless Pauli
factor at one of these sites, so
$\bra{\psi_a}\rbra*{I\otimes P}\ket{\psi_a}=0$.
Denote by $\sigma^{\sbra{R}}$ the state obtained from $\sigma$
by tracing out the system and reference qubits outside $R$.
Apply the same partial trace to
$\sum_{a:R\subseteq S_a}w_a\ketbra{\psi_a}{\psi_a}$
and denote the resulting operator by $B_P$.
Then $0\leq B_P\leq\sigma^{\sbra{R}}$ and
$\Tr\rbra{B_P}=w\rbra{P}$.
The terms with $R\nsubseteq S_a$ contribute zero, giving
\begin{equation}
\label{eq:lo-local-coefficient}
    a_P
    =\Tr\rbra*{\rbra*{I\otimes P\rvert_R}\sigma^{\sbra{R}}}
    =\Tr\rbra*{\rbra*{I\otimes P\rvert_R}B_P}.
\end{equation}

The joint space of the system and reference qubits
at the sites in $R$ has dimension $4^{\abs{R}}$.
Use its Pauli-vector basis $\cbra{\opket{Q}}_Q$,
where $Q$ ranges over all Pauli strings on $R$.
Here $\opket{Q}$ uses normalized vectorization with
factor $2^{-\abs{R}/2}$.
Normalized vectorization gives
\begin{align*}
    \rbra*{I\otimes P\rvert_R}\opket{Q}
    =\opket{Q\rbra*{P\rvert_R}^{\mathsf T}}.
\end{align*}
Denote by $\pi_P\rbra{Q}$ the Pauli label of
$Q\rbra*{P\rvert_R}^{\mathsf T}$, with the overall
phase disregarded.
The map $\pi_P$ permutes the Pauli labels on $R$.
Since $P$ has a nonidentity factor at every site in $R$,
$Q$ and $\pi_P\rbra{Q}$ cannot both have an identity factor
at any such site.
Therefore
\begin{equation}
\label{eq:lo-pair-cover}
  \supp\rbra*{Q}\cup\supp\rbra*{\pi_P\rbra*{Q}}=R.
\end{equation}
Denote by $p_Q$ the probability of outcome $Q$
when measuring $\sigma^{\sbra{R}}$ in this basis.
The global Pauli-vector measurement acts separately
on the system--reference qubit pair at each site.
Hence, for each $i=1,\ldots,m$,
\begin{align*}
    p_Q
    =\opbra{Q}\sigma^{\sbra{R}}\opket{Q}
    =\Pr\sbra*{Q_i\rvert_R=Q}.
\end{align*}
For Pauli labels $Q,Q'$ on $R$, write
$\rbra{B_P}_{Q,Q'}=\opbra{Q}B_P\opket{Q'}$
for the matrix entries of $B_P$ in this basis.
Since $0\leq B_P\leq\sigma^{\sbra{R}}$ and
$\Tr\rbra{B_P}=w\rbra{P}$, each diagonal entry satisfies
$0\leq\rbra{B_P}_{Q,Q}\leq\min\cbra*{p_Q,w\rbra{P}}$.
Positivity also gives
\begin{align*}
    \abs*{\rbra{B_P}_{Q,\pi_P\rbra{Q}}}^2
    \leq\rbra{B_P}_{Q,Q}\rbra{B_P}_{\pi_P\rbra{Q},\pi_P\rbra{Q}}.
\end{align*}
Applying the two diagonal bounds to opposite factors yields
\begin{equation}
\label{eq:lo-entry-bounds}
    \abs*{\rbra{B_P}_{Q,\pi_P\rbra{Q}}}^2\leq w\rbra{P}p_Q,
    \qquad
    \abs*{\rbra{B_P}_{Q,\pi_P\rbra{Q}}}^2\leq w\rbra{P}p_{\pi_P\rbra{Q}}.
\end{equation}

If both $Q$ and $\pi_P\rbra{Q}$ appear among the restricted
outcomes $Q_1\rvert_R,\ldots,Q_m\rvert_R$, then the
supports of the corresponding global outcomes cover
$R$ by~\eqref{eq:lo-pair-cover}, so $P\in\calC$.
Thus $P\notin\calC$ implies that at least one of these
two local labels is absent.
Independence of the measurements and the union bound
give, for every $Q$,
\begin{equation}
\label{eq:lo-absence}
  \Pr\sbra*{P\notin\calC}
    \le \rbra{1-p_Q}^m+\rbra{1-p_{\pi_P\rbra{Q}}}^m
    \le e^{-mp_Q}+e^{-mp_{\pi_P\rbra{Q}}}.
\end{equation}
Expanding the trace in
\eqref{eq:lo-local-coefficient} in the local
Pauli-vector basis expresses $a_P$ as a sum of
$4^{\abs{R}}$ entries $\rbra{B_P}_{Q,\pi_P\rbra{Q}}$, each multiplied
by a phase of modulus one.
The triangle and Cauchy--Schwarz inequalities give
\begin{align}
    a_P^2
    \leq\rbra*{\sum_Q\abs*{\rbra{B_P}_{Q,\pi_P\rbra{Q}}}}^2
    \leq 4^{\abs{R}}\sum_Q\abs*{\rbra{B_P}_{Q,\pi_P\rbra{Q}}}^2.
\end{align}
Combining \eqref{eq:lo-absence} with
\eqref{eq:lo-entry-bounds} and using
$x e^{-mx}\leq1/\rbra{em}$ for $x\geq0$ yields
\begin{align}
    a_P^2\Pr\sbra*{P\notin\calC}
    &\leq 4^{\abs{R}}\sum_Q\abs*{\rbra{B_P}_{Q,\pi_P\rbra{Q}}}^2
        \rbra*{e^{-mp_Q}+e^{-mp_{\pi_P\rbra{Q}}}}\\
    &\leq 4^{\abs{R}} w\rbra{P}\sum_Q
        \rbra*{p_Qe^{-mp_Q}
        +p_{\pi_P\rbra{Q}}e^{-mp_{\pi_P\rbra{Q}}}}\\
    &\leq\frac{2\cdot16^{\abs{R}}}{em}\,w\rbra{P}.
\end{align}
If $\operatorname{wt}\rbra{P}>k$, then
$\supp\rbra{P}\nsubseteq S_a$ for every $a$,
so $a_P=0$.
For the remaining terms,
$16^{\abs{R}}\leq16^k$.
Summing over $P\in\calP_n^\circ$ and using
\eqref{eq:lo-support-weight-sum} gives
\begin{align}
    \mathbb E\sbra*{G-G_{\calC}}
    &=
    \sum_{P\in\calP_n^\circ}
    a_P^2\Pr\sbra*{P\notin\calC}\\
    &\leq
    \frac{2\cdot16^k}{em}
    \sum_{P\in\calP_n^\circ}w\rbra{P}\\
    &\leq
    \frac{2\cdot16^k\rbra{4^k-1}}{em}.
\end{align}
\end{proof}

\subsection{Copy upper bound}
\label{sec:copy-upper-bound}

Fix the first-batch outcomes and the resulting set
$\calC$.
Let $S=\bigcup_{j=1}^m\supp\rbra*{Q_j}$, so every
candidate is supported on $S$ and $\abs{S}\leq km$.
On each fresh copy of $\sigma$, discard the system
register.
For each $i\in S$, choose $B_i$ uniformly from
$\cbra{X,Y,Z}$, independently across sites and copies,
and measure reference qubit $i$ in the eigenbasis
of $B_i$.
Denote the resulting eigenvalue by
$b_i\in\cbra{-1,1}$.
For $P\in\calC$, use the standard product-Pauli record~\cite{HuangKuengPreskill2020}:
\[
    X_P=3^{\operatorname{wt}\rbra{P}}
        \prod_{i\in\supp\rbra*{P}}
        b_i\,\mathbf 1\cbra*{B_i=P_i}.
\]
Let $\vec{X}=\rbra{X_P}_{P\in\calC}$ be the resulting
classical record vector, and let
$\vec{a}=\rbra{a_P}_{P\in\calC}$.

Let $N\geq2$ be an integer.
Apply the measurements described above to $N$
independent fresh copies of $\sigma$ to obtain
records $\vec{X}_1,\ldots,\vec{X}_N$, and set
\[
    \widehat G
    =\frac{1}{N\rbra{N-1}}
      \sum_{\substack{1\leq r,s\leq N\\r\ne s}}
      \vec{X}_r^{\mathsf T}\vec{X}_s
    =\frac{
        \Abs*{\sum_{r=1}^N\vec{X}_r}_2^2
        -\sum_{r=1}^N\Abs{\vec{X}_r}_2^2
      }{N\rbra{N-1}}.
\]
This is a classical $U$-statistic of order two.
It is computed from the stored classical records,
without joint measurements across the $N$ copies.
If $\calC=\varnothing$, the records are
zero-dimensional vectors and $\widehat G=0$.

The next lemma gives the conditional expectation and a variance
bound for $\widehat G$.

\begin{lemma}[Conditional variance]
\label{lem:lo-variance}
Conditional on the first-batch outcomes
$Q_1,\ldots,Q_m$, the estimator $\widehat G$ satisfies
\begin{align}
    \mathbb E\sbra*{\widehat G\mid Q_1,\ldots,Q_m}
    &=G_{\calC},\\
    \operatorname{Var}
        \sbra*{\widehat G\mid Q_1,\ldots,Q_m}
    &\leq
      \frac{4\cdot6^k\rbra{2^k-1}}{N}
      +\frac{4\cdot16^k\cdot6^{2k}m^2}{N^2}.
\end{align}
\end{lemma}

\begin{proof}
All expectations and covariances in this proof are
conditional on the first-batch outcomes.
Denote by $\Sigma$ the covariance matrix of $\vec{X}$.
For $P\in\calC$, the measurement axes match $P$ on
its support with probability $3^{-\operatorname{wt}\rbra{P}}$.
Conditional on a match, the product
$\prod_{i\in\supp\rbra*{P}}b_i$ has expectation
$\Tr\rbra*{P\sigma_{\mathsf R}}=a_P$.
Since $X_P=0$ otherwise,
\[
    \mathbb E\sbra*{X_P}
    =3^{\operatorname{wt}\rbra{P}}
      \cdot3^{-\operatorname{wt}\rbra{P}}a_P
    =a_P.
\]
Independence of the fresh records gives, for $r\ne s$,
\[
    \mathbb E\sbra*{\vec{X}_r^{\mathsf T}\vec{X}_s}
    =\vec{a}^{\mathsf T}\vec{a}
    =G_{\calC}.
\]
Averaging over the $N\rbra{N-1}$ ordered pairs gives
$\mathbb E\sbra*{\widehat G}=G_{\calC}$.

Denote by $\mathbb E_{\mathrm{mix}}$ expectation over
the measurement axes and outcomes for the same
measurements applied to $I_S/2^{\abs{S}}$.
Under this state, independent uniform outcomes and
axis choices give
\[
  \mathbb E_{\mathrm{mix}}\sbra*{X_PX_Q}
   =\mathbf 1\cbra*{P=Q}\,3^{\operatorname{wt}\rbra{P}}.
\]
For distinct strings, either their nonidentity axes conflict at some
site, or an outcome sign occurs only once and averages to zero.
Consequently,
$\mathbb E_{\mathrm{mix}}\sbra*{\vec{X}\vec{X}^{\mathsf T}}\leq3^k I_{\abs{\calC}}$.

By \cref{lem:jump-state-locality-bounds},
$\sigma_{\mathsf R}\leq2^kI/d$.
Taking the partial trace over $S^c$ gives
\[
  \Tr_{S^c}\rbra*{\sigma_{\mathsf R}}
  \leq2^k\frac{I_S}{2^{\abs{S}}}.
\]
For each fixed choice of axes, taking traces against
the corresponding positive measurement effects bounds
each outcome probability by $2^k$ times its value for
$I_S/2^{\abs{S}}$.
The axes have the same distribution for both states,
so the same bound holds for the joint probabilities
of axes and outcomes.
Consequently, for each real vector $\vec{u}\in\mathbb R^{\abs{\calC}}$,
\[
  \mathbb E\sbra*{\rbra*{\vec{u}^{\mathsf T}\vec{X}}^2}
    \le2^k\mathbb E_{\mathrm{mix}}\sbra*{\rbra*{\vec{u}^{\mathsf T}\vec{X}}^2}
    \le6^k\Abs{\vec{u}}_2^2.
\]
It follows that $\mathbb E\sbra*{\vec{X}\vec{X}^{\mathsf T}}\leq6^k I_{\abs{\calC}}$.
Subtracting $\vec{a}\vec{a}^{\mathsf T}\geq0$ gives
$\Sigma\leq6^k I_{\abs{\calC}}$.

The bound $\sigma_{\mathsf R}\leq2^kI/d$ also gives
\begin{align}
  G_{\calC}
  \leq G
  =d\Tr\rbra*{\sigma_{\mathsf R}^2}-1
  \leq 2^k\Tr\rbra*{\sigma_{\mathsf R}}-1
  =2^k-1.
\end{align}

Write $\vec{Y}_r=\vec{X}_r-\vec{a}$ for $1\leq r\leq N$.
Then $\mathbb E\sbra*{\vec{Y}_r}=0$ and
$\mathbb E\sbra*{\vec{Y}_r\vec{Y}_r^{\mathsf T}}=\Sigma$.
Substituting $\vec{X}_r=\vec{a}+\vec{Y}_r$ into $\widehat G$
gives an ordered-pair sum in which each linear term
$\vec{a}^{\mathsf T}\vec{Y}_r$ occurs $2\rbra{N-1}$ times.
Hence
\[
  \widehat G-G_{\calC}
    =\frac2N\sum_{r=1}^N \vec{a}^{\mathsf T}\vec{Y}_r
      +\frac1{N\rbra{N-1}}
       \sum_{\substack{1\leq r,s\leq N\\r\ne s}}
       \vec{Y}_r^{\mathsf T}\vec{Y}_s.
\]
The two sums are uncorrelated because every cross term contains an
independent zero-mean record appearing only once. Distinct unordered
pairs in the second sum are likewise uncorrelated: if they share an
index, conditioning on the shared record leaves the other records
independent and centered. By independence,
\[
    \mathbb E\sbra*{
        \rbra*{\vec{Y}_1^{\mathsf T}\vec{Y}_2}^2}
    =\mathbb E\sbra*{\vec{Y}_1^{\mathsf T}\Sigma\vec{Y}_1}
    =\Tr\rbra*{
        \Sigma\mathbb E\sbra*{
            \vec{Y}_1\vec{Y}_1^{\mathsf T}}}
    =\Tr\rbra*{\Sigma^2}.
\]
The linear term has variance
$4\vec{a}^{\mathsf T}\Sigma\vec{a}/N$.
There are $N\rbra{N-1}/2$ unordered pairs, each occurring
twice in the quadratic sum.
The variance of the quadratic term is therefore
\[
    \frac{4}{N^2\rbra{N-1}^2}\cdot\frac{N\rbra{N-1}}2
    \Tr\rbra*{\Sigma^2}
    =\frac{2}{N\rbra{N-1}}\Tr\rbra*{\Sigma^2}.
\]
Since $\Sigma\leq6^k I_{\abs{\calC}}$,
$\vec{a}^{\mathsf T}\Sigma\vec{a}\leq6^kG_{\calC}$ and
$\Tr\rbra*{\Sigma^2}\leq6^{2k}\abs{\calC}$.
Using these bounds and $G_{\calC}\leq2^k-1$ gives
\begin{align}
  \operatorname{Var}\sbra*{\widehat G\mid Q_1,\ldots,Q_m}
    &=\frac4N\vec{a}^{\mathsf T}\Sigma\vec{a}
      +\frac{2}{N\rbra{N-1}}\Tr\rbra*{\Sigma^2}\\
    &\leq\frac{4\cdot6^kG_{\calC}}N
      +\frac{2\cdot6^{2k}\abs{\calC}}{N\rbra{N-1}}\\
    &\leq\frac{4\cdot6^k\rbra{2^k-1}}N
      +\frac{2\cdot6^{2k}\abs{\calC}}{N\rbra{N-1}}.
\end{align}
Since $\abs{\calC}\leq16^km^2$ and $N-1\geq N/2$,
\begin{align}
    \operatorname{Var}\sbra*{\widehat G\mid Q_1,\ldots,Q_m}
    &\leq
    \frac{4\cdot6^k\rbra{2^k-1}}{N}
    +\frac{4\cdot16^k\cdot6^{2k}m^2}{N^2}.
\end{align}
\end{proof}

\Cref{lem:lo-discovery} bounds
$\mathbb E\sbra*{G-G_{\calC}}$, while
\cref{lem:lo-variance} shows that $\widehat G$ is
conditionally unbiased for $G_{\calC}$ and bounds
its conditional variance.
Since $h\rbra{\sigma}=\Tr\rbra*{\sigma^2}+G/2$,
it remains to estimate the purity $\Tr\rbra*{\sigma^2}$.

\begin{lemma}[Purity estimation, adapted from {\cite{EkertEtAl2002}}]
\label{lem:purity-estimation}
Fix $\eta>0$.
Let $\rho$ be a quantum state and let $s\geq1$ be an integer.
Using $2s$ independent copies of $\rho$, one can obtain
an estimator $\widehat\mu$ of $\Tr\rbra*{\rho^2}$ satisfying
\[
    \Pr\sbra*{
        \abs*{\widehat\mu-\Tr\rbra*{\rho^2}}>\eta
    }
    \leq2e^{-s\eta^2/2}.
\]
\end{lemma}

The estimator averages the outcomes of $s$ SWAP tests
on disjoint pairs of copies, with ancilla outcomes
$0,1$ encoded as $+1,-1$, respectively.
Each encoded outcome has expectation
$\Tr\rbra*{\rho^2}$.
The stated bound follows from Hoeffding's inequality.

Combining these estimates yields the following upper
bound on the copy complexity of estimating
$h\rbra{\sigma}$.

\begin{theorem}[Copy upper bound]
\label{thm:jump-state-estimation}
Fix a locality bound $k\geq1$.
Let $n\geq1$ be an integer and let $0<\eta\leq1/16$.
There exists an algorithm that, given independent copies
of the jump state $\sigma$ of any nonzero $n$-qubit
canonical dissipator with locality at most $k$,
outputs $\widehat h\in\sbra*{0,\frac{2^k+1}{2}}$
satisfying
$\abs*{\widehat h-h\rbra*{\sigma}}\leq\eta$
with probability at least $0.99$.
The algorithm uses $O\rbra{\eta^{-2}}$ copies.
\end{theorem}

\begin{proof}
Set $\beta=1/300$ and $K_k=2\cdot16^k\rbra{4^k-1}/e$.
Choose $m$, $N$, and $s$ before any measurements:
\begin{align*}
    m=\ceil*{\frac{2K_k}{\beta\eta}}, \qquad
    N =\ceil*{\max\cbra*{
        \frac{32\cdot6^k\rbra{2^k-1}}{\beta\eta^2},
        \frac{8\cdot6^k\cdot4^km}{\sqrt{\beta}\eta}
    }}, \qquad
    s =\ceil*{8\eta^{-2}\log\rbra*{\frac{2}{\beta}}}.
\end{align*}

Use three disjoint batches of $m$, $N$, and $2s$
independent copies of $\sigma$.
Measure the first batch in the Pauli-vector basis,
obtaining $Q_1,\ldots,Q_m$.
If any outcome has weight greater than $k$, return $\widehat h=0$.
Otherwise, construct from $Q_1,\ldots,Q_m$ the candidate
set $\calC$ used in \cref{lem:lo-discovery}.
Use the second batch to compute the U-statistic
$\widehat G$ analyzed in \cref{lem:lo-variance}
from independent product-Pauli measurement records.
Under the locality assumption, every first-batch outcome
has weight at most $k$ with probability one.
For arbitrary input states, this check ensures
$\abs{\calC}\leq16^km^2$ whenever the algorithm continues.

Let $\mu=\Tr\rbra*{\sigma^2}$.
Apply \cref{lem:purity-estimation} to the third batch
with $\rho=\sigma$ and accuracy $\eta/2$ to obtain
$\widehat\mu$.
Each SWAP exchanges the full $\mathsf S\mathsf R$
registers between the two copies.
Return
\[
    \widehat h
    =\min\cbra*{
        \frac{2^k+1}{2},
        \max\cbra*{0,\widehat\mu+\frac{\widehat G}{2}}
    }.
\]

By \cref{lem:jump-state-locality-bounds},
$0\leq h\rbra{\sigma}\leq\rbra{2^k+1}/2$.
Projecting $\widehat\mu+\widehat G/2$ onto this interval
cannot increase the absolute estimation error.

Since $G-G_{\calC}\geq0$, Markov's inequality and
\cref{lem:lo-discovery} give
\[
  \Pr\sbra*{G-G_{\calC}>\eta/2}
  \leq\frac{2\mathbb E\sbra*{G-G_{\calC}}}{\eta}
  \leq\frac{2K_k}{m\eta}
  \leq\beta.
\]

Conditional on the first-batch outcomes,
\cref{lem:lo-variance} gives
$\mathbb E\sbra*{\widehat G\mid Q_1,\ldots,Q_m}=G_{\calC}$.
The variance bound in \cref{lem:lo-variance}
and the choice of $N$ give
\begin{align}
  \operatorname{Var}\sbra*{\widehat G\mid Q_1,\ldots,Q_m}
  &\leq\frac{4\cdot6^k\rbra{2^k-1}}{N}
       +\frac{4\cdot16^k\cdot6^{2k}m^2}{N^2}\\
  &\leq\frac{\beta\eta^2}{8}+\frac{\beta\eta^2}{16}
   =\frac{3\beta\eta^2}{16}.
\end{align}
Chebyshev's inequality gives
\begin{align}
  \Pr\sbra*{
    \abs*{\widehat G-G_{\calC}}>\eta/2
    \mid Q_1,\ldots,Q_m
  }
  \leq
  \frac{4\operatorname{Var}\sbra*{
    \widehat G\mid Q_1,\ldots,Q_m
  }}{\eta^2}
  \leq\frac{3\beta}{4}\leq\beta.
\end{align}
Averaging over the first-batch outcomes gives the same
unconditional bound.

By \cref{lem:purity-estimation} and the choice of $s$,
\[
  \Pr\sbra*{\abs*{\widehat\mu-\mu}>\eta/2}
  \leq2e^{-s\eta^2/8}
  \leq\beta.
\]

By the union bound, with probability at least $1-3\beta=0.99$,
the quantities $G-G_{\calC}$,
$\abs*{\widehat G-G_{\calC}}$, and $\abs*{\widehat\mu-\mu}$
are all at most $\eta/2$.
On this event, using $h\rbra*{\sigma}=\mu+G/2$ gives
\begin{align}
  \abs*{\widehat h-h\rbra*{\sigma}}
  &\leq
  \abs*{\widehat\mu+\widehat G/2-h\rbra*{\sigma}}\\
  &\leq
  \abs*{\widehat\mu-\mu}
  +\frac12\abs*{\widehat G-G_{\calC}}
  +\frac12\rbra*{G-G_{\calC}}\\
  &\leq
  \frac{\eta}{2}+\frac{\eta}{4}+\frac{\eta}{4}
  =\eta.
\end{align}

The algorithm uses at most $m+N+2s$ copies.
Since $k$ and $\beta$ are fixed,
the batch sizes satisfy $m=O\rbra{\eta^{-1}}$,
$N=O\rbra{\eta^{-2}}$, and $s=O\rbra{\eta^{-2}}$.
Thus the total copy complexity is $O\rbra{\eta^{-2}}$.
\end{proof}

\subsection{Copy lower bound}
\label{sec:copy-lower-bound}
This subsection shows that the dependence on $\eta$
in the copy upper bound of
\cref{thm:jump-state-estimation} is optimal for fixed $k$.
The $\Omega\rbra*{\eta^{-2}}$ lower bound for purity
estimation~\cite[Theorem~3.2]{ChenLiuWang2026TracePowers}
can be transferred to estimating $h$ by embedding
the two-dimensional states used in that proof into
mixed jump states generated by two jump operators
acting on the same qubit.
These states have maximally mixed reference marginals,
so $h\rbra*{\sigma}=\Tr\rbra*{\sigma^2}$.
The construction below establishes the same lower bound
even for a single jump operator acting on one qubit.
The corresponding jump states are pure, so
$\Tr\rbra*{\sigma^2}=1$ and the variation of $h$
comes entirely from their reference marginals.

The lower-bound argument uses the following standard bound.

\begin{lemma}[Holevo--Helstrom bound, {\cite{Helstrom1967,Holevo1973}}]
\label{lem:holevo-helstrom}
Let $\rho_0$ and $\rho_1$ be quantum states on the same
finite-dimensional Hilbert space.
The optimal success probability for distinguishing them,
given one copy and equal prior probabilities, is $\frac12+\frac14\Abs*{\rho_0-\rho_1}_1$.
\end{lemma}

\begin{theorem}[Copy lower bound]
\label{thm:jump-state-copy-lower}
Fix a locality bound $k\geq1$.
Let $n\geq1$ be an integer and let $0<\eta\leq1/16$.
Let $\sigma$ range over the jump states of nonzero
$n$-qubit canonical dissipators with locality at most $k$.
Any estimator that accesses $\sigma$ only through
independent copies and, for every such $\sigma$,
returns an estimate $\widehat h$ satisfying
$\abs*{\widehat h-h\rbra*{\sigma}}\leq\eta$
with probability at least $2/3$ requires
$\Omega\rbra*{\eta^{-2}}$ copies.
\end{theorem}

The lower bound holds even for estimators using arbitrary
adaptive and collective measurements.
It also holds when inputs are restricted to jump states
generated by a single jump operator acting on one qubit.

\begin{proof}
Suppose the estimator uses at most $N$ copies,
where $N\geq0$ is an integer.
Let $X_1,Y_1,Z_1$ denote the Pauli operators on the
first qubit, tensored with the identity on the remaining
qubits.
For $\theta\in\rbra{0,\pi/4}$, consider the canonical
dissipator with the single traceless jump operator
\[
  L_\theta=\cos\theta\,X_1+i\sin\theta\,Z_1.
\]
Since $L_\theta$ acts nontrivially only on the first qubit,
the dissipator has locality at most $k$ for every $k\geq1$.
The Pauli identities $X_1^2=Z_1^2=I$,
$X_1Z_1=-iY_1$, and $Z_1X_1=iY_1$ give
\[
  A_\theta=L_\theta^\dagger L_\theta
  =I+i\cos\theta\sin\theta\rbra*{X_1Z_1-Z_1X_1}
  =I+\sin\rbra*{2\theta}Y_1.
\]
Since $\Tr\rbra*{Y_1}=0$, the total jump rate satisfies
\[
  \Gamma_\theta=\overline{\Tr}\rbra*{A_\theta}=1.
\]

Since $\opbra{L_\theta}\opket{L_\theta}=\Gamma_\theta=1$,
the corresponding jump state is pure:
\[
  \sigma_\theta=\opket{L_\theta}\opbra{L_\theta},
  \qquad
  \Tr\rbra*{\sigma_\theta^2}=1.
\]
Taking the system partial trace gives
\[
  \rbra*{\sigma_\theta}_{\mathsf R}
  =\Tr_{\mathsf S}\rbra*{
    \opket{L_\theta}\opbra{L_\theta}}
  =\frac{\rbra*{L_\theta^\dagger L_\theta}^{\mathsf T}}{d}
  =\frac{A_\theta^{\mathsf T}}{d}.
\]
Using $Y_1^2=I$ and $\Tr\rbra*{Y_1}=0$ then gives
\[
  d\Tr\rbra*{\rbra*{\sigma_\theta}_{\mathsf R}^2}
  =\overline{\Tr}\rbra*{A_\theta^2}
  =1+\sin^2\rbra*{2\theta}.
\]
Consequently,
\begin{equation}
  h\rbra*{\sigma_\theta}
  =\Tr\rbra*{\sigma_\theta^2}
   +\frac12\rbra*{
     d\Tr\rbra*{\rbra*{\sigma_\theta}_{\mathsf R}^2}-1}
  =1+\frac12\sin^2\rbra*{2\theta}.
  \label{eq:copy-lower-functional}
\end{equation}

Set $\theta_0=\pi/8$ and $\theta_1=\pi/8+4\eta$.
The assumption $0<\eta\leq1/16$ ensures
$0<\theta_0<\theta_1<\pi/4$ and
$0<16\eta\leq1<\pi/2$.
Using \cref{eq:copy-lower-functional} and
$\sin x\geq2x/\pi$ for $x\in\sbra{0,\pi/2}$ gives
\begin{align}
  h\rbra*{\sigma_{\theta_1}}
  -h\rbra*{\sigma_{\theta_0}}
  &=\frac12\rbra*{
    \sin^2\rbra*{\frac{\pi}{4}+8\eta}-\frac12}\\
  &=\frac14\sin\rbra*{16\eta}\\
  &\geq\frac{8\eta}{\pi}>2\eta.
\end{align}

Compare the estimate $\widehat h$ with the midpoint of
the two functional values.
Identify the input as $\sigma_{\theta_0}$ if
$\widehat h$ is below the midpoint, and as
$\sigma_{\theta_1}$ otherwise.
Since the functional values differ by more than $2\eta$,
the decision is correct whenever the estimation error
is at most $\eta$.
Thus the success probability is at least $2/3$
on either input.

Supply all $N$ copies at the start and discard any
unused copies.
The entire algorithm, including its internal randomness
and the final decision, corresponds to a two-outcome
measurement on $\sigma_\theta^{\otimes N}$.
For equal prior probabilities, \cref{lem:holevo-helstrom}
gives the optimal success probability
$\frac12+\frac14\Abs*{
\sigma_{\theta_1}^{\otimes N}
-\sigma_{\theta_0}^{\otimes N}}_1$.
The success probability of at least $2/3$ therefore gives
\begin{align}
  \frac13
    &\leq\frac12\Abs*{
      \sigma_{\theta_1}^{\otimes N}
      -\sigma_{\theta_0}^{\otimes N}}_1.
  \label{eq:copy-lower-trace-separation}
\end{align}

Pauli orthonormality gives
\[
  \opbra{L_{\theta_0}}\opket{L_{\theta_1}}
  =\cos\theta_0\cos\theta_1+\sin\theta_0\sin\theta_1
  =\cos\rbra*{4\eta}.
\]
Let
\[
  B=\sigma_{\theta_1}^{\otimes N}-\sigma_{\theta_0}^{\otimes N}.
\]
By \cref{eq:copy-lower-trace-separation}, $B\ne0$.
The operator $B$ is Hermitian, has trace zero, and has rank
at most two.
Its two nonzero eigenvalues are therefore $\lambda$ and
$-\lambda$ for some $\lambda>0$, so
\[
  \frac12\Abs{B}_1=\lambda,
  \qquad
  \Tr\rbra*{B^2}=2\lambda^2.
\]
Both tensor-power states are rank-one projectors.
Expanding the square and using cyclicity of the trace gives
\begin{align}
  \Tr\rbra*{B^2}
  &=\Tr\rbra*{\rbra*{\sigma_{\theta_1}^{\otimes N}}^2}
    +\Tr\rbra*{\rbra*{\sigma_{\theta_0}^{\otimes N}}^2}
    -2\Tr\rbra*{
      \sigma_{\theta_1}^{\otimes N}\sigma_{\theta_0}^{\otimes N}}\\
  &=2-2\rbra*{\Tr\rbra*{
      \sigma_{\theta_1}\sigma_{\theta_0}}}^N\\
  &=2-2\abs*{\opbra{L_{\theta_0}}\opket{L_{\theta_1}}}^{2N}\\
  &=2\rbra*{1-\cos^{2N}\rbra*{4\eta}}.
\end{align}
Using $1-u^N\leq N\rbra{1-u}$ for $u\in\sbra{0,1}$
and $\sin x\leq x$ for $x\geq0$ therefore gives
\begin{align}
  \frac12\Abs{B}_1
  &=\sqrt{\frac{\Tr\rbra*{B^2}}{2}}\\
  &=\sqrt{1-\cos^{2N}\rbra*{4\eta}}\\
  &\leq\sqrt{N\sin^2\rbra*{4\eta}}\\
  &\leq4\eta\sqrt N.
\end{align}
Combining this with~\cref{eq:copy-lower-trace-separation}
and squaring yields $N\geq1/\rbra{144\eta^2}$, proving the claimed
lower bound.
\end{proof}

\section{Time-optimal tolerant testing}
\label{sec:locality-only}

This section establishes the optimal total time for
the testing problem and oracle model specified in
\cref{sec:testing-strategy-overview}.
\Cref{subsec:lo-time-upper,subsec:lo-time-lower}
establish matching time upper and lower
bounds, respectively, for fixed $k$.

\subsection{Time upper bound}
\label{subsec:lo-time-upper}

The upper bound combines estimates of the total jump rate
and $h\rbra{\sigma}$, using the following flagged state
preparation and the estimator in
\cref{thm:jump-state-estimation}.
For $t>0$, prepare $\ket{\Phi}$, apply $\calE_t$
to the system register, and perform the two-outcome
projective measurement
$\cbra{\ketbra{\Phi}{\Phi},\Pi}$
on the resulting Choi state $J\rbra{\calE_t}$.
The measurement records only the binary outcome
and preserves coherence within the success subspace
$\ket{\Phi}^{\perp}$.
On outcome $\Pi$, the preparation succeeds and outputs
the postmeasurement state, whose unnormalized
density operator is $B_t=\Pi J\rbra{\calE_t}\Pi$.
Otherwise, replace the quantum output with a flag
$\ket{\mathrm{fail}}$ orthogonal to the success subspace.
Each preparation uses total time $t$.

The output state, including the flag, is
\[
  \Omega_t
  =\rbra*{1-\Tr\rbra*{B_t}}
     \ketbra{\mathrm{fail}}{\mathrm{fail}}
     \oplus B_t.
\]

Recall that $C=\Pi J\rbra*{\calL}\Pi$,
$\Gamma=\Tr\rbra*{C}$, and
$\sigma=C/\Gamma$ when $\Gamma>0$.
For $\Lambda t\leq1/2$, compare $\Omega_t$
with the ideal flagged state
\begin{equation}
\label{eq:lo-ideal-flag}
  \Omega_t^{\mathrm{ideal}}
  =\rbra*{1-\Gamma t}\ketbra{\mathrm{fail}}{\mathrm{fail}}
    \oplus tC.
\end{equation}
Since $C\geq0$, $\Tr\rbra*{C}=\Gamma$, and
$0\leq\Gamma t\leq\Lambda t\leq1/2$
by \cref{lem:generator-strength-bounds}, this is a quantum state.
For $\Gamma>0$, the ideal preparation succeeds with
probability $\Gamma t$ and outputs $\sigma$
conditional on success.
When $\Gamma=0$, its output is the failure flag.

\begin{theorem}[Time upper bound]
\label{thm:lo-time}
Fix a locality bound $k\geq1$.
Let $n\geq1$ be an integer, let $\Lambda>0$ be a strength bound,
and let $0\leq\varepsilon_1<\varepsilon_2\leq\Lambda/\sqrt2$
be thresholds.
Let $\calL$ be an unknown time-independent $n$-qubit
Lindblad generator with $\Abs{\calL}_\diamond\leq\Lambda$
whose canonical dissipator $\calD$ has locality at most $k$.
There is a tester that distinguishes $\Abs{\calD}_F\leq\varepsilon_1$ from $\Abs{\calD}_F\geq\varepsilon_2$
with success probability at least $2/3$, using total time
\[
  O\rbra*{
    \frac{\varepsilon_2}
         {\rbra*{\varepsilon_2-\varepsilon_1}^2}}.
\]
\end{theorem}

\begin{proof}
Write $\delta=\varepsilon_2-\varepsilon_1$.
Set $q_k=\sqrt{1+\rbra{2^k+1}/2}$.
When $\Gamma>0$, \cref{lem:jump-state-locality-bounds}
gives $0\leq h\rbra{\sigma}\leq\rbra{2^k+1}/2$.
Together with \cref{lem:rate-state-identity}, this yields
\[
  \Gamma
  \leq\Abs{\calD}_F
  =\Gamma\sqrt{1+h\rbra{\sigma}}
  \leq q_k\Gamma.
\]
When $\Gamma=0$, $\calD=0$, so the two inequalities
also hold.

Following the accuracy choice motivated in
\cref{sec:testing-strategy-overview}, set
$\eta=\delta/\rbra{32\varepsilon_2}$,
so $0<\eta\leq1/32$.
Choose $C_k\geq1$, depending only on $k$, sufficiently large
that the algorithm constructed in the proof of
\cref{thm:jump-state-estimation}, with accuracy $\eta$,
uses at most $\ceil*{C_k\eta^{-2}}$ copies.
Choose the deterministic parameters
\begin{align*}
  T=\frac{1000q_k^2}{\varepsilon_2}\ceil*{C_k\eta^{-2}}, \qquad
  M=\ceil*{\max\cbra*{1,2\Lambda T,100e\Lambda^2T^2}},
  \qquad t=\frac TM.
\end{align*}
The choice of $M$ gives $\Lambda t=\Lambda T/M\leq1/2$,
so the ideal flagged output is a state.
Run exactly $M$ independent flagged preparations,
each producing $\Omega_t$.
Let $K$ be the total number of successes, and store
the first $\min\cbra*{K,\ceil*{C_k\eta^{-2}}}$
successful outputs.
Set $\widehat\Gamma=K/T$.
Use the following decision rule:
\begin{enumerate}
\item If $\widehat\Gamma\le\varepsilon_2/\rbra{2q_k}$, return
$\mathrm{YES}$.
\item If $\widehat\Gamma\ge2\varepsilon_2$, return $\mathrm{NO}$.
\item Otherwise, apply the algorithm constructed in the
proof of \cref{thm:jump-state-estimation}, with accuracy
$\eta$, to the stored outputs, obtaining $\widehat h$.
Return $\mathrm{YES}$ if
$\widehat\Gamma\sqrt{1+\widehat h}
 \leq\rbra{\varepsilon_1+\varepsilon_2}/2$,
and $\mathrm{NO}$ otherwise.
\end{enumerate}
There are enough copies to apply the third decision rule,
since its condition implies
\[
  K>\frac{T\varepsilon_2}{2q_k}
   =500q_k\ceil*{C_k\eta^{-2}}
   \geq\ceil*{C_k\eta^{-2}}.
\]
The tester makes exactly $M$ oracle calls of duration $t$,
using total time $Mt=T$.
The call count is $M=O\rbra*{1+\Lambda T+\Lambda^2T^2}$;
this is separate from the total-evolution-time resource bound.

First analyze the ideal experiment.
If $\Gamma=0$, every preparation fails and the tester
returns $\mathrm{YES}$.
For the remaining ideal analysis, suppose $\Gamma>0$.
Then $K\sim\operatorname{Bin}\rbra{M,\Gamma t}$, so
\begin{equation}
\label{eq:lo-rate-moments}
  \mathbb E\sbra*{\widehat\Gamma}=\Gamma,
  \qquad
  \operatorname{Var}\sbra*{\widehat\Gamma}
  =\frac{\Gamma\rbra*{1-\Gamma t}}{T}
  \leq\frac{\Gamma}{T}.
\end{equation}
Suppose $\Gamma\le4\varepsilon_2$. Chebyshev's inequality gives
\begin{align*}
  \Pr\sbra*{\abs*{\widehat\Gamma-\Gamma}>
                        \frac{\delta}{16q_k}}
  \le\frac{1024q_k^2\varepsilon_2}{T\delta^2}
  \le\frac1{1000},
\end{align*}
where the second inequality uses
$\ceil*{C_k\eta^{-2}}\geq\eta^{-2}$.
On the complementary event, if the first decision rule applies, then
\[
  \Abs{\calD}_F\le q_k\Gamma
    \le\frac{\varepsilon_2}{2}+\frac\delta{16}
    <\varepsilon_2.
\]
Thus returning $\mathrm{YES}$ is correct.
If the second decision rule applies, then
\[
  \Abs{\calD}_F\ge\Gamma
     \ge2\varepsilon_2-\frac\delta{16q_k}>\varepsilon_1,
\]
so returning $\mathrm{NO}$ is also correct.

Conditional on any flag sequence for which the third decision rule applies,
the stored outputs are independent copies of $\sigma$.
Thus \cref{thm:jump-state-estimation} gives
$\abs*{\widehat h-h\rbra*{\sigma}}\leq\eta$
with conditional failure probability at most $0.01$.
The estimator always returns
$\widehat h\in\sbra*{0,\rbra{2^k+1}/2}$.
On the event that both estimation errors satisfy their stated bounds,
using $\abs{\sqrt{1+x}-\sqrt{1+y}}\leq\abs{x-y}/2$
for $x,y\geq0$ gives
\begin{align}
  \abs*{\widehat\Gamma\sqrt{1+\widehat h}
      -\Abs{\calD}_F}
  &\leq
    \abs*{\widehat\Gamma-\Gamma}\sqrt{1+\widehat h}
    +\Gamma\abs*{\sqrt{1+\widehat h}
                  -\sqrt{1+h\rbra{\sigma}}}\\
  &\leq
    q_k\abs*{\widehat\Gamma-\Gamma}
    +\frac{\Gamma}{2}
       \abs*{\widehat h-h\rbra{\sigma}}\\
  &\leq\frac{\delta}{16}+2\varepsilon_2\eta
   =\frac{\delta}{8}.
\end{align}
Since $\delta/8<\delta/2$, the midpoint decision
is correct on this event.
Averaging the conditional estimation failure probability
over flag sequences for which the third decision rule applies and using the
union bound gives total failure probability at most
$1/1000+0.01=0.011$ when $\Gamma\leq4\varepsilon_2$.

If $\Gamma>4\varepsilon_2$, the input is in the
$\mathrm{NO}$ case.
The tester is correct whenever
$\widehat\Gamma\geq2\varepsilon_2$.
By Chebyshev's inequality and
\cref{eq:lo-rate-moments},
\begin{align}
  \Pr\sbra*{\widehat\Gamma<2\varepsilon_2}
  \leq
    \Pr\sbra*{\abs*{\widehat\Gamma-\Gamma}
                 \geq\Gamma/2}
  \leq\frac{4}{\Gamma T}
  \leq\frac{1}{\varepsilon_2T}
   =\frac{1}{1000q_k^2\ceil*{C_k\eta^{-2}}}
   \leq\frac{1}{1000}.
\end{align}
Thus the ideal tester has failure probability
at most $0.011$.

Finally, compare the physical and ideal experiments.
The Taylor remainder and $\Lambda t\leq1/2$ give
\begin{align}
  \Abs*{e^{t\calL}-\Id_{\calH}-t\calL}_\diamond
  &\leq e^{\Lambda t}-1-\Lambda t\\
  &\leq\frac12\Lambda^2t^2e^{\Lambda t}\\
  &\leq\frac e2\Lambda^2t^2.
\end{align}
Since $\Pi J\rbra*{\Id_{\calH}}\Pi=0$ and
$\Pi J\rbra*{\calL}\Pi=C$, compression by $\Pi$ and
the normalized Choi representation give
\begin{align}
  \Abs*{B_t-tC}_1
  &=\Abs*{\Pi J\rbra*{
      e^{t\calL}-\Id_{\calH}-t\calL}\Pi}_1\\
  &\leq\Abs*{e^{t\calL}-\Id_{\calH}-t\calL}_\diamond\\
  &\leq\frac e2\Lambda^2t^2.
\end{align}
The direct-sum structure of the flagged outputs then gives
\begin{align}
  \frac12\Abs*{\Omega_t-\Omega_t^{\mathrm{ideal}}}_1
  &=\frac12\rbra*{
      \abs*{\Tr\rbra*{B_t-tC}}+\Abs*{B_t-tC}_1}\\
  &\leq\Abs*{B_t-tC}_1.
\end{align}
Telescoping the tensor product and using $Mt=T$ yields
\begin{align}
  \frac12\Abs*{\Omega_t^{\otimes M}
      -\rbra*{\Omega_t^{\mathrm{ideal}}}^{\otimes M}}_1
  &\leq\frac M2\Abs*{\Omega_t-\Omega_t^{\mathrm{ideal}}}_1\\
  &\leq\frac e2M\Lambda^2t^2\\
  &=\frac{e\Lambda^2T^2}{2M}\\
  &\leq\frac1{200}.
\end{align}
The last inequality follows from $M\geq100e\Lambda^2T^2$.
Both experiments apply the same estimator and decision
rule to their flagged outputs.
By contractivity of trace distance under this processing,
their probabilities of returning either answer differ
by at most $1/200$.
Thus the physical tester has failure probability at most
$0.011+1/200=0.016<1/3$ and succeeds with
probability greater than $2/3$. The total time satisfies
\[
  T
  =\frac{1000q_k^2}{\varepsilon_2}\ceil*{C_k\eta^{-2}}
  =O\rbra*{\frac{1}{\varepsilon_2\eta^2}}
  =O\rbra*{\frac{\varepsilon_2}{\delta^2}}.
\]
This proves the theorem.
\end{proof}

\subsection{Time lower bound}
\label{subsec:lo-time-lower}

The following theorem gives a time lower bound that
matches the upper bound in \cref{thm:lo-time}
for fixed $k$.
The lower bound already holds for generators with
zero Hamiltonian part and dissipators of locality
at most one.
We use the dephasing family from the learning-time lower bound
in~\cite[Theorem~4.2 and Corollary~4.3]{AradEtAl2026},
and compare full Poisson records to obtain the two-threshold testing bound.

\begin{theorem}[Time lower bound]
\label{thm:lo-time-lower}
Fix a locality bound $k\geq1$.
Let $n\geq1$ be an integer, let $\Lambda>0$ be a strength bound,
and let $0\leq\varepsilon_1<\varepsilon_2\leq\Lambda/\sqrt2$
be thresholds.
Let $\calL$ range over time-independent $n$-qubit
Lindblad generators with $\Abs{\calL}_\diamond\leq\Lambda$
whose canonical dissipators $\calD$ have locality
at most $k$.
Any tester that distinguishes
$\Abs{\calD}_F\leq\varepsilon_1$ from
$\Abs{\calD}_F\geq\varepsilon_2$
with success probability at least $2/3$
requires total time
\[
  \Omega\rbra*{
    \frac{\varepsilon_2}
         {\rbra*{\varepsilon_2-\varepsilon_1}^2}}.
\]
\end{theorem}

\begin{proof}
Suppose the tester uses total time at most $T$.
Write $\delta=\varepsilon_2-\varepsilon_1$.
Let $Z_1$ denote the Pauli $Z$ operator on the first qubit, tensored
with the identity on the remaining qubits. For $\Gamma\geq0$, consider
the single-qubit dephasing family
\[
  \calL_\Gamma\rbra*{\rho}=\Gamma\rbra*{Z_1\rho Z_1-\rho}.
\]
The generator has zero Hamiltonian part and the traceless
jump operator $\sqrt{\Gamma}\,Z_1$.
Thus $\calD_\Gamma=\calL_\Gamma$, and its locality
is at most one.
Its Choi operator is
\[
  J\rbra{\calL_\Gamma}
  =\Gamma\rbra*{
    \opket{Z_1}\opbra{Z_1}
    -\opket{I}\opbra{I}}.
\]
The two Pauli vectors are orthonormal, so
\cref{fact:choi-frobenius} gives
\[
  \Abs{\calD_\Gamma}_F
  =\Abs{J\rbra{\calL_\Gamma}}_2
  =\sqrt2\,\Gamma.
\]
Moreover,
\[
  2\Gamma
  =\Abs{J\rbra{\calL_\Gamma}}_1
  \leq\Abs{\calL_\Gamma}_\diamond
  \leq2\Gamma,
\]
where the last inequality follows from the triangle
inequality and the unit diamond norm of each unitary channel.
Set
\[
  \Gamma_0=\frac{\varepsilon_1}{\sqrt2},\qquad
  \Gamma_1=\frac{\varepsilon_2}{\sqrt2}.
\]
The generators $\calL_{\Gamma_0}$ and $\calL_{\Gamma_1}$
lie at the two promise thresholds and satisfy the
strength bound, since
$2\Gamma_0\leq2\Gamma_1=\sqrt2\,\varepsilon_2\leq\Lambda$.
The case $\Gamma_0=0$ is allowed.

A call of duration $t$ can be simulated by a rate-$\Gamma$ Poisson
process of phase flips, because
\[
  e^{t\calL_\Gamma}\rbra*{\rho}
    =e^{-\Gamma t}\sum_{r=0}^{\infty}
       \frac{\rbra{\Gamma t}^r}{r!}Z_1^r\rho Z_1^r.
\]
For a tester with total time at most $T$, use a Poisson record on
a cumulative time interval $\sbra{0,T}$. Each call consumes
its chosen interval and applies the flips recorded there.
Given the consumed record and the tester's classical
history, subsequent measurement outcomes can be sampled
using auxiliary randomness independent of the Poisson
process, with conditional probabilities that do not
depend on $\Gamma$.
Each call duration is chosen using the previous outcomes
and internal random choices.
The cumulative call boundaries are therefore stopping
times for the Poisson process together with this
independent randomness.
The strong Markov property gives the correct increment
for each adaptive call.
Parallel calls act on distinct
registers, so their channels commute even on entangled inputs.
Placing these calls sequentially on the cumulative time axis
preserves both their joint channel and the sum of their durations.

Consequently, a stochastic map independent of $\Gamma$
transforms the full Poisson record into the tester's output.
By data processing, the total variation distance between
the two output distributions is at most that between
the corresponding Poisson record distributions.

For $i\in\cbra{0,1}$, let $p_i$ be the probability
mass function of $\operatorname{Pois}\rbra{\Gamma_iT}$.
Conditional on an event count $r$ of positive probability,
the ordered event times have the distribution of the order
statistics of $r$ independent uniform points in $\sbra{0,T}$.
This conditional distribution is independent of $\Gamma$.
When $\Gamma_0=0$, the record is empty almost surely.
The count is a function of the record, and the record
can be generated from the count using this common
conditional distribution.
Data processing in both directions therefore shows that
the total variation distance between the record distributions
equals $\dtv{p_0}{p_1}$.
Moreover,
\begin{align}
  \affh{p_0}{p_1}
  &=\sum_{r=0}^{\infty}
       \sqrt{e^{-\Gamma_0T}\frac{\rbra{\Gamma_0T}^r}{r!}
             e^{-\Gamma_1T}\frac{\rbra{\Gamma_1T}^r}{r!}}\\
  &=e^{-\rbra{\Gamma_0+\Gamma_1}T/2}
    \sum_{r=0}^{\infty}
    \frac{\rbra{T\sqrt{\Gamma_0\Gamma_1}}^r}{r!}\\
  &=\exp\rbra*{-\frac T2
       \rbra{\sqrt{\Gamma_1}-\sqrt{\Gamma_0}}^2}.
\end{align}
Factoring the difference of the two probability mass
functions using square roots and applying Cauchy--Schwarz gives
\begin{align}
  \dtv{p_0}{p_1}
  &\leq\frac12\sqrt{
      \rbra{2-2\affh{p_0}{p_1}}
      \rbra{2+2\affh{p_0}{p_1}}}\\
  &=\sqrt{1-\affh{p_0}{p_1}^2}.
\end{align}
The probability of returning $\mathrm{YES}$ is at least
$2/3$ under $\calL_{\Gamma_0}$ and at most $1/3$ under
$\calL_{\Gamma_1}$.
The output distributions therefore differ by at least $1/3$
on this event, so data processing gives $\dtv{p_0}{p_1}\geq1/3$.
Combining the bounds gives
\[
  e^{-T\rbra{\sqrt{\Gamma_1}-\sqrt{\Gamma_0}}^2}
  =\affh{p_0}{p_1}^2\leq\frac89.
\]
Taking logarithms and substituting the two rates yields
\begin{align}
  T
  &\geq\frac{\log\rbra{9/8}}
       {\rbra{\sqrt{\Gamma_1}-\sqrt{\Gamma_0}}^2}\\
  &=\frac{\sqrt2\log\rbra{9/8}}
       {\rbra{\sqrt{\varepsilon_2}-\sqrt{\varepsilon_1}}^2}\\
  &=\frac{\sqrt2\log\rbra{9/8}
       \rbra{\sqrt{\varepsilon_2}+\sqrt{\varepsilon_1}}^2}
       {\delta^2}\\
  &\geq\sqrt2\log\rbra{9/8}
       \frac{\varepsilon_2}{\delta^2},
\end{align}
which proves the theorem.
\end{proof}

\paragraph{Acknowledgment}
J.~B. is supported by the Quantum Advantage Turbo Charger (QATCH)
programme, funded by UKRI under Grant No.~UKRI4257, and by the
Quantum Advantage Pathfinder (QAP) programme, funded by EPSRC under
Grant No.~EP/X026167/1.
Y.~C. is supported in part by the Air Force Office of Scientific Research under grant agreement FA9550-21-1-0392, by the Stanford Graduate Fellowship in Science \& Engineering, and by the National Science Foundation Graduate Research Fellowship Program under Grant No. DGE-2146755.
Z.~S. is supported by the European Research Council (ERC) grant under Grant No.~GIFNEQ 101163938.

\paragraph{AI use statement}
The core ideas of this work originated with the authors. The main technical methods were developed through iterative exchanges between the authors and OpenAI’s GPT-5.6 Sol and GPT-6 Astra, accessed through Codex. These models also assisted with language editing, manuscript organization, notation consistency, literature searches, and the review of mathematical arguments and derivations. The authors reviewed and revised the AI-assisted content and independently verified the mathematical results and references. The authors take full responsibility for the accuracy and integrity of the manuscript.

\bibliographystyle{alphaurl}
\phantomsection
\addcontentsline{toc}{section}{References}
\bibliography{main}

\end{document}